\documentclass[11pt,oneside, reqno]{amsart}
\usepackage[applemac]{inputenc}

\usepackage{amsbsy}
\usepackage{amstext}
\usepackage{amssymb}
\usepackage{graphicx}
\usepackage{setspace}
\usepackage{url}
\usepackage{esint}
\usepackage{multirow}
\usepackage{blkarray}
\usepackage[pdftex,bookmarks=true,colorlinks=true,citecolor=blue,plainpages=false]{hyperref}
\usepackage[round]{natbib}
\PassOptionsToPackage{normalem}{ulem}
\usepackage{ulem}

\usepackage{enumitem}

\usepackage{dsfont}

\setlist[enumerate,1]{label={\arabic*.}}

\usepackage{amsthm}
\usepackage{color}
\usepackage{bm}

\author{Angelo Mele}
\address{Carey Business School\\ 
Johns Hopkins University\\
100 International Dr. \\
Baltimore, MD 21202 }

\author{Lingxin Hao}
\address{Department of Sociology\\
Johns Hopkins University\\
3400 North Charles St. \\
Baltimore, MD 21218}

\author{Joshua Cape} 
\address{Department of Statistics\\
University of Pittsburgh\\
230 South Bouquet Street \\ 
Pittsburgh, PA 15260} 

\author{Carey E. Priebe} 
\address{Department of Applied Mathematics and Statistics\\
Johns Hopkins University \\
3400 North Charles St. \\
Baltimore, MD 21218}

\title[Stochastic blockmodels with covariates]{Spectral estimation of large stochastic blockmodels with discrete nodal covariates}

\thanks{This version: \today. First version: February 25, 2019. Contact: \href{mailto:angelo.mele@jhu.edu}{angelo.mele@jhu.edu}; \href{mailto:hao@jhu.edu}{hao@jhu.edu}; \href{mailto:joshua.cape@pitt.edu}{joshua.cape@pitt.edu}; \href{mailto:cep@jhu.edu}{cep@jhu.edu}. We are grateful to Cong Mu and Jipeng Zhang for excellent research assistance. We thank Eric Auerbach, Federico Bandi, Stephane Bonhomme, Youngser Park and Eleonora Patacchini for comments and suggestions. Funding from the Institute of Data Intensive Engineering and Science (IDIES) at Johns Hopkins University and NSF grant SES-1951005 is gratefully acknowledged. Joshua Cape also gratefully acknowledges support from NSF grant DMS-1902755.}

\newtheorem{theorem}{THEOREM}

\newtheorem{lemma}{LEMMA}

\makeatother

\begin{document}
\begin{abstract}
In many applications of network analysis, it is important to distinguish between observed and unobserved factors affecting network structure. To this end, we develop spectral estimators for both unobserved blocks and the effect of covariates in stochastic blockmodels. 
On the theoretical side, we establish asymptotic normality of our estimators for the subsequent purpose of performing inference. On the applied side, we show that computing our estimator is much faster than standard variational expectation--maximization algorithms and scales well for large networks. Monte Carlo experiments suggest that the estimator performs well under different data generating processes. Our application to Facebook data shows evidence of homophily in gender, role and campus-residence, while allowing us to discover unobserved communities. The results in this paper provide a foundation for spectral estimation of the effect of observed covariates as well as unobserved latent community structure on the probability of link formation in networks.
\end{abstract}


\maketitle

\onehalfspacing
\section{Introduction}
\subsection{Motivation}
The analysis and modeling of network data has far-reaching applications in economics, sociology, public health, computer science, neuroscience, and marketing, among other areas. Social networks have been shown to affect socioeconomic performance, such as education \citep{CalvoArmengolEtAl2009, DeGiorgiPellizzariRedaelli2009, CarrellEtAl2013}, health and risky behaviors \citep{Nakajima2007, Badev2013}, risk sharing arrangements \citep{FafchampsGubert2006}, and employment opportunities \citep{Topa2001, Beaman2012}, among others. 
Often, both observed and unobserved factors contribute to the global structure of networks and the processes that generate them. For example, in social networks, factors including gender, race, and personality affect the likelihood that two people interact. In applications, race and gender are typically observed, while personality is usually unobserved. It is therefore crucial to develop ways to disentangle the effect of observed and unobserved variables on link formation in networks.

The analysis of network models poses several econometric challenges, arising from the structure of correlation among links \citep{Chandrasekhar2016, DePaula2017, Graham2020, GrahamDePaula2020, Dzemski2017}. In this paper, we focus on models with conditionally independent links, where the unobserved heterogeneity is modeled by latent positions in low-dimensional Euclidean space \citep{Graham2014, Auerbach2019}. While this formulation rules out externalities and strategic considerations in link formation, which are dominant focal points in much of the econometric literature on strategic network formation games \citep{Mele2017, DePaula2017, Graham2020, GrahamDePaula2020, Menzel2017}, there is a growing literature showing how models with conditionally independent links provide useful approximations of strategic models, at least in some special cases \citep{Mele2017, MeleZhu2020, DiaconisChatterjee2011, Graham2020}.

In this paper, we analyze the stochastic blockmodel (SBM), a workhorse in the literature on community detection and clustering \citep{Abbe2018, NowickiSniders2011, AiroldiEtAl2008}. In standard $K$-block SBMs, each node belongs to one of $K$ unobserved blocks (communities); conditional on the block assignments, links form independently as Bernoulli variables with probabilities that depend on the community memberships. Many existing works use SBMs to estimate or approximate unobserved block structure in networks. In contrast, applications involving SBMs that incorporate observed nodal attributes as covariates are comparatively few \citep{Sweet2015, ChoiEtAl2011, AtchadeEtAl2019}. A possible reason is that estimation for stochastic blockmodels is computationally burdensome, and including covariates in the specification imposes significant additional challenges to modelling, estimation, and inference. Exact maximum likelihood estimation is infeasible, causing most estimation strategies to rely on approximations based on expectation--maximization algorithms and variational methods \citep{AiroldiEtAl2008,DaudinEtAl2008,BickelEtAl2013,LatoucheEtAl2012,Vu2013}. However, these algorithms may converge slowly to the (approximate) solution and become impractical for networks with thousands of nodes.\footnote{Recent advances use further approximations and parallelization to improve computational efficiency \citep{AtchadeEtAl2019,Vu2013}. We do not pursue such extensions in this paper.}

\subsection{Overview of the contributions}

The goal of our paper is to develop an estimation method for SBMs with observed nodal covariates that is computationally feasible and scales well to large networks, while providing similar statistical guarantees as the other available methods. At a high level, our main strategy consists of formulating this goal as an inference problem in the context of the generalized random dot product graph (GRDPG) model \citep{AthreyaEtAl2018a,TangPriebe2018, TangEtAl2017, Rubin-DelanchyEtAl2018, AlidaeeEtAl2020, MuEtAl2021}. In a GRDPG, each node is characterized by an unobserved latent position (vector), and each pair of nodes links with probability determined via a (possibly indefinite) inner product of the pair's latent positions; crucially, any SBM can be reformulated as a GRDPG where the latent positions are fixed within blocks. To address our computational goals, we turn to the use of spectral methods, which, in addition to enjoying theoretical guarantees \citep{TangEtAl2017}, have been shown to be successful both in terms of feasibility and scalability in related settings \citep{PriebeFlashGraph, PriebeSSD2016,PriebeBillionNodesGraphs}.

We provide several contributions to the literature on network econometrics. First, we present a GRDPG model framework that incorporates the effect of observed covariates on link probabilities. Second, we develop a spectral estimator for inference in stochastic blockmodels with covariates, adapting the spectral estimators developed for our new class of GRDPG models. Crucially, we obtain a new central limit theorem for the spectral estimator of the covariates' effect. Our estimator is asymptotically normal as long as the parameter(s) for covariate effects can be written as sufficiently well-behaved functions of the SBM block-specific probabilities. We provide explicit formulas for bias and variance properties of the estimator, and we show that the estimator is computationally fast, scaling well for large networks. Our method provides a statistical and algorithmic foundation for inference in a broad class of models for large network data, including networks that are relatively sparse in the sense that their average degree scales sub-linearly with network size.

Our exposition focuses on SBMs with a single binary (or discrete) observed covariate, though we emphasize that the theoretical and computational properties set forth in this work extend to settings involving multiple discrete covariates. Asymptotic normality continues to hold as long as the estimator for the effect of the covariate(s) can be expressed as a suitably well-behaved function of the SBM probabilities. We illustrate several examples with two binary covariates in a simulation study and showcase our method in an empirical application using Facebook data and three control variables. The case involving continuous covariates is more complicated and an active area of contemporaneous research. Current progress on this front is being facilitated by recently investigated Latent Structure Models (LSM) \citep{AthreyaLSM2018} and related ideas.

The development of our estimator depends crucially on several observations. First, a $K$-block stochastic blockmodel with one binary covariate can be reformulated as a (different yet related) $2K$-block stochastic blockmodel. Second, as discussed previously, a stochastic blockmodel graph can be viewed as a generalized random dot product graph whose latent positions are fixed within blocks \citep{AthreyaEtAl2018a,TangPriebe2018,TangEtAl2017}. The behavior of our spectral estimation method is tied to the asymptotic behavior of spectral estimators for SBM block probability matrix entries recently studied in \cite{TangEtAl2017}.
Our asymptotic analysis provides explicit formulas for standard errors and establishes the existence of a bias term; however, this bias vanishes at a rate proportional to the size of the network.\footnote{Since we have a closed-form expression for the bias term, in principle we can naively correct for it in estimation, using a plug-in estimate. In our simulations we find that the bias term is usually so small that the correction is not necessary, at least for networks with a few thousand nodes. On the other hand, the bias is demonstrably substantial in the empirical application to Facebook data in Section~\ref{section:simulations}.}

The theoretical machinery used to perform inference extends methods developed for the analysis of latent positions network models \citep{AthreyaEtAl2018a,TangEtAl2017}. In particular, we use Adjacency Spectral Embedding (ASE) for random graphs to embed the network in a low-dimensional space and to recover the latent positions of the nodes. Our method is motivated by the (verifiable) intuition that the adjacency matrix can be viewed as a (mild) perturbation of the probability matrix that generates the network data, and thus, that the eigenstructure of the adjacency matrix resembles that of the edge probability matrix \citep{TangPriebe2018,AthreyaEtAl2018a}. In particular, spectrally decomposing the adjacency matrix provides accurate information about the structure of sufficiently large networks \citep{TangEtAl2017}.

In addition to providing statistical guarantees, one of the advantages of our method is the speed of computation, obtained without sacrificing estimation accuracy. In our simulations (see Section~\ref{section:simulations}) we compare our approach to the variational EM (VEM) algorithm \citep{DaudinEtAl2008, BickelEtAl2013}, as implemented in the \texttt{blockmodels} package in \texttt{R}. Even for the simplest case of a stochastic blockmodel without covariates, our spectral method is faster by several orders of magnitude. For example, in a network with $n=5000$ nodes and $K=2$ blocks, we can estimate the model in few seconds using our spectral method, while it takes almost 10 minutes to estimate the model using the variational EM algorithm. When we add a binary covariate, our estimator converges in under 30 seconds, while in contrast it takes almost 10 hours when using a parallelized version of the VEM algorithm in \texttt{blockmodels}. Our methods are implemented in the R package \texttt{grdpg} available at \url{https://github.com/meleangelo/grdpg} and all replication files can be found at \url{https://github.com/meleangelo/grdpg_supplement}.

We perform a Monte Carlo study to examine the performance of our spectral estimator. 
We simulate networks in which the block structure is unbalanced and blocks have different sizes; we include scenarios in which the covariates are binary and independent, as well as cases in which the covariates are correlated. Our results confirm the existence of bias in small samples, but the bias decreases in magnitude when the size of the network increases.
As expected, the bias is larger when the blocks are unbalanced and the covariates are correlated. These insights confirm that our estimator works best in very large networks, where the bias problem is attenuated, thus adding to the advantage of computational speed.

Finally, we apply our method to the study of Facebook friendship data using the Facebook~100 dataset, initially collected and analyzed in \cite{TraudEtAl2012}. These data contain networks of friendships and node (person) information for 100 universities in the United States in the year 2005. We estimate a stochastic blockmodel for the Harvard University network, consisting of more than 13,000 nodes, using information on gender, off-campus residence, and university role (i.e.,~student, faculty, staff, etc.) of the users (see also \cite{AtchadeEtAl2019} for a related analysis).
We find evidence of homophily, as suggested by the positive effect of gender, role, and off-campus residence on the probability of linking. This suggests that including information about observable covariates in the estimation may allow researchers to better recover unobservable block structure.

Some of the theoretical machinery presented herein will be useful to econometricians studying networks, as well as nonlinear panel data. Indeed, some of the GRDPG modeling and inference ideas are similar to the literature on interactive fixed effects  \citep{ChenWeidnerFernandezVal2021, FERNANDEZVAL2016, MoonShumWeidner2018}.\footnote{An application of SBM-related ideas is \cite{Bonhomme2017}.}
Another application of our model is to correct for the endogeneity of the network in empirical models of network effects \citep{ShaliziMcFowland2018,GoldsmithPinkhamImbens2013, BoucherFortin2016,Auerbach2019, JohnssonMoon2019}. Currently most of these studies rely on an auxiliary model of network formation to capture unobserved heterogeneity that affects the outcome. Our model and computational method allow the researcher to perform this type of correction for large network data.

\section{Background and methodology}

\subsection{Stochastic blockmodels and generalized random dot product graphs} 
In a $K$-block \emph{Stochastic Blockmodel (SBM)}, nodes are randomly assigned to one of $K$ blocks; conditional on the blocks, nodes form links independently. A $K$-block SBM is characterized by the $K\times K$ matrix of probabilities $\bm{\theta}\in [0,1]^{K\times K}$, where  the entry $\bm{\theta}_{k\ell}$ is the probability of a link occurring between nodes in blocks $k$ and $\ell$. The random variables comprising $\bm{\tau}=(\tau_1, \dots, \tau_n)$ describe the assignments of each node to a block, and they are i.i.d., such that the probability that node $i$ belongs to block $k$ is $\mathbb{P}(\tau_i=k)=\pi_k$, with $\bm{\pi} = (\pi_1, \dots, \pi_K)$.
Conditional on the assignment to blocks $\bm{\tau}$, the probability that nodes $i$ and $j$ have a link is $\bm{P}_{ij}=\bm{\theta}_{\tau_i \tau_j}$. We use the $n\times n$ adjacency matrix $\bm{A}$ to describe the network, conditioning on the unobserved blocks. According to the SBM the entries of the adjacency matrix are generated as
\begin{equation}
\bm{A}_{ij} \vert \tau_i, \tau_j,\overset{ind}{\sim} Bernoulli(\bm{P}_{ij} ),
\end{equation}
and we write $ (\bm{A}, \bm{\tau}) \sim SBM(\bm{\theta},\bm{\pi})$ to denote the adjacency matrix drawn from a $K$-block SBM with probability matrix $\bm{\theta}$ and block assignment probabilities $\bm{\pi}$.\\

Conditioning on $\bm{\tau}$, the likelihood of the SBM with $K$ blocks is
\begin{eqnarray}
    \mathbb{P}(\bm{A}\vert \bm{\tau}) &=& \prod_{i\leq j}\bm{P}_{ij}^{\bm{A}_{ij}}(1- \bm{P}_{ij})^{1-\bm{A}_{ij}} 
    = \prod_{i\leq j}\bm{\theta}_{\tau_i \tau_j}^{\bm{A}_{ij}}(1- \bm{\theta}_{\tau_i \tau_j})^{1-\bm{A}_{ij}}.
\end{eqnarray}

The \emph{Generalized Random Dot Product Graph} (GRPDG) model is an alternative model for network formation with conditionally independent links. 
In a GRDPG, each node $i$ is characterized by a $d$-dimensional vector (i.e., an \emph{unobserved} latent position) $\bm{X}_{i}= (X_{i1},\dots,X_{id})\in \mathcal{X}_d\subseteq \mathbb{R}^d$. The latent positions are i.i.d.~draws from a distribution $F$ with support $\mathcal{X}_d$, that is
$\bm{X}_1, \bm{X}_2, \dots, \bm{X}_n \overset{iid}{\sim} F$.
Let $\bm{X}=[\bm{X}_1, \dots, \bm{X}_n]^T$ denote the matrix formed by row-wise stacking all unobserved vectors $\bm{X}_{i}$. 

Let $d_1 \ge 1$ and $d_2 \ge 0$ be integers, and define $d = d_1+d_2$. Let $\bm{I}_{d_1, d_2}$ be a $d\times d$ diagonal matrix containing $1$'s in  $d_1$ diagonal entries and $-1$ in the remaining $d_2$ diagonal entries. For a \emph{GRDPG with signature $(d_1,d_2)$}, the entries of the adjacency matrix $\bm{A}_{ij}$ are specified to be independent, after conditioning on the latent positions $\bm{X}_i$ and $\bm{X}_j$, namely

\begin{equation}
\bm{A}_{ij} \vert \bm{X}_i, \bm{X}_j \overset{ind}{\sim} Bernoulli(\bm{P}_{ij}),
\end{equation}
with link probability given by
\begin{equation}
    \bm{P}_{ij} = \bm{X}_i^T \bm{I}_{d_{1},d_{2}}\bm{X}_j.
\end{equation}
For this setting, we write $(\bm{X},\bm{A}) \sim GRDPG_{d_1,d_2}(F)$.\footnote{
    It must be noted that the support $\mathcal{\bm{X}}_d$ of $F$, is a subset of $\mathbb{R}^d$ such that $\bm{x}^T \bm{I}_{d_1,d_2} \bm{y}\in[0,1]$ for all $\bm{x},\bm{y}\in \mathcal{\bm{X}}_d$.
    }\\

\subsection{The relationship between SBMs and GRDPGs}
A notable property of GRDPGs is that they encompass or approximate any conditionally independent low-rank network model. In particular, \emph{any SBM can be represented as a GRDPG} with latent positions fixed within blocks. That is, the $K$ blocks are represented by a fixed location, so that each $\bm{X}_i$ can only take values $\bm{\nu} = \left[\bm{\nu}_1, \bm{\nu}_2, \dots, \bm{\nu}_K\right]$. Two nodes $i$ and $j$ belong to the same block $k$ if $\bm{X}_i=\bm{X}_j=\bm{\nu}_k$. The random variables $\bm{\tau}$ are such that $\tau_1,\dots,\tau_n \overset{iid}{\sim} Multinomial(1;\pi_1, \dots,\pi_K)$, with $\bm{\pi} \in (0,1)^K$ and $\sum_{k=1}^{K} \pi_k =1$.\footnote{Alternatively, we can think of a $K$-block stochastic blockmodel as a network where the $\bm{X}_i$'s are drawn from a mixture of degenerate distributions with mass centered at $\bm{\nu}$, i.e.,
\begin{equation}
    \bm{X}_i \sim \pi_1 \delta_{\bm{\nu}_1} + \pi_2 \delta_{\bm{\nu}_2} + \dots + \pi_K \delta_{\bm{\nu}_K}.
\end{equation}}
The GRDPG corresponding to model $(\bm{A}, \bm{\tau}) \sim SBM(\bm{\theta},\bm{\pi})$ can be obtained by an eigendecomposition of the matrix $\bm{\theta} = \bm{U}\bm{\Sigma} \bm{U}^T$ and by defining $\bm{\nu}_1,\bm{\nu}_2, \dots,\bm{\nu}_K$ as the rows of $\bm{U}\vert \bm{\Sigma}\vert^{1/2}$. The distribution $F$ is $F=\sum_{k=1}^{K}\pi_k \delta_{\bm{\nu}_k}$, where $\delta$ is the Dirac-delta; importantly, $d$ is the rank of the block-probabilities matrix $\bm{\theta}$, and $d_1, d_2$ are the number of positive and negative eigenvalues of matrix $\bm{\theta}$, respectively.\\

This paper utilizes inferential spectral methods for GRDPGs to estimate SBMs \citep{AthreyaEtAl2018a,TangEtAl2017}. The same relationship between SBMs and GRDPGs holds for \emph{known} link functions and $\bm{\theta}_{\tau_i\tau_j} = h\left(\bm{B}_{\tau_i\tau_j} \right)$, where $h$ is a known function that maps to $[0,1]$ and $\bm{B}$ is a $K\times K$ matrix of real numbers.\footnote{If $h$ is unknown we cannot in general expect to be able to accurately estimate the latent positions. See \cite{tang2013}.} For example, $h$ could be the logistic function or the cumulative density function of the Gaussian distribution. Our stochastic blockmodel would have adjacency matrix $\bm{A}$ with elements
\begin{equation}
    \bm{A}_{ij}\vert \tau_i, \tau_j \overset{ind}{\sim} Bernoulli\left(h(\bm{B}_{\tau_i \tau_j}) \right).
\end{equation}

The stochastic blockmodel can be extended to include the effect of observed covariates \citep{ChoiEtAl2011, Sweet2015, AtchadeEtAl2019}. Such models allow researchers to disentangle the effect of observed and unobserved nodal heterogeneity on the probability of linking. In particular, in social science, such models are used to estimate to what extent the network exhibits homophily or heterophily.  
Let node $i$ be characterized by an $r$-dimensional vector of \emph{observed} covariates $\bm{Z}_{i}= (Z_{i}^{(1)}, \dots,Z_{i}^{(r)}) \in \mathcal{Z}\subseteq \mathbb{R}^r$ and let the stochastic blockmodel be
\begin{equation}
\bm{A}_{ij} \vert \tau_i, \tau_j, \bm{Z}_i, \bm{Z}_j\overset{ind}{\sim} Bernoulli\left( h (\bm{B}_{\tau_i \tau_j} + f(\bm{Z}_i, \bm{Z}_j;\bm{\beta}) ) \right),
\end{equation}
where $f$ is a known function, $\bm{\beta}$ is a vector of parameters, and where we allow $\bm{Z}_i$ to (possibly) depend on the latent blocks.\\

In this paper \emph{we will focus on the case of a single binary (or discrete) covariate $\bm{Z}_i$, and we will assume that the function $f$ is an indicator variables that indicates whether $i$ and $j$'s covariates have the same value}, i.e., 
\begin{equation}
 f(\bm{Z}_i, \bm{Z}_j;\bm{\beta}) = \beta \bm{1}_{\lbrace \bm{Z}_{i}=\bm{Z}_{j}\rbrace}.   
\end{equation}
Here $\beta$ can be interpreted in terms of homophily. Namely, if $\beta>0$, the probability of a link between $i$ and $j$ is higher when their observables $\bm{Z}_{i}$ and $\bm{Z}_{j}$ are the same, and the network displays homophily in the observable variable. Viceversa, when $\beta<0$, we have heterophily. The extension to multiple discrete covariates has similar properties and will be discussed further below.\\

Our goal is to develop a \emph{general spectral method of inference} for the parameter $\beta$ and for $\bm{B}_{\tau_i \tau_j}$ in the following stochastic blockmodel with a discrete nodal covariate:
\begin{eqnarray}
\bm{A}_{ij} \vert \tau_i, \tau_j, \bm{Z}_i, \bm{Z}_j &\overset{ind}{\sim}& Bernoulli\left( \bm{P}_{ij} \right),\\
\bm{P}_{ij} &=& h \left(\bm{B}_{\tau_i \tau_j} + \beta \bm{1}_{\lbrace \bm{Z}_{i}=\bm{Z}_{j}\rbrace} \right).
\end{eqnarray}
We further wish to disentangle the effect of observed and unobserved heterogeneity on link probabilities. To achieve this, we need to extend results from previous work on GRDPGs and SBMs \citep{AthreyaEtAl2018a,TangPriebe2018,TangEtAl2017}. In the following subsections we review some of the spectral methods we use in the paper, and we provide an example that highlights the core aspects of our method.

\subsection{Spectral methods and spectral embeddings}
Estimation of SBMs for large networks, with or without observed covariates, is computationally challenging. The exact MLE problem is intractable because of the high-dimensional combinatorial problem of considering all possible partitions of the nodes in blocks \citep{BickelEtAl2013}. Approximate methods are available, based on variational approximations \citep{DaudinEtAl2008,AiroldiEtAl2008,WainwrightJordan2008}; however, even these methods are computationally prohibitive for large networks.\\
We make use of spectral methods, which have been shown in the literature to scale well with network size. Our spectral approach embeds the network into a low(er) dimensional space, thus reducing the dimensionality of the problem, while maintaining the geometric properties of the data. In particular we use the \emph{Adjacency Spectral Embedding} (ASE) to estimate the latent positions of the GRDPG \citep{AthreyaEtAl2018a}. In this sense, our method can be considered a dimension-reduction tool that decreases the complexity of the data by reducing the dimensionality of the space. 
The intuition about the spectral method is that if $\bm{P}$ is a low-rank matrix, then we can think of the adjacency matrix $\bm{A}$ as a perturbation of $\bm{P}$, that is $\bm{A}_{ij} = \bm{P}_{ij}+ \bm{E}_{ij}$, where $\bm{E}_{ij}$ is a matrix of independent stochastic perturbations.\footnote{In the Bernoulli case, $\bm{E}_{ij}$ is a shifted Bernoulli variable, with values $\bm{E}_{ij}=1-\bm{P}_{ij}$ with probability $\bm{P}_{ij}$ and $\bm{E}_{ij}=\bm{P}_{ij}$ with probability $1-\bm{P}_{ij}$.} If $\bm{A}$ and $\bm{P}$ are \emph{close enough}, namely if $\bm{E}$ is \emph{small enough}, then the leading eigenvalues and eigenvectors of $\bm{A}$ and $\bm{P}$ will be similar \citep{TangEtAl2017}. As a consequence, the spectral decomposition of $\bm{A}$ will provide an estimate of the latent structure of the network, that is, the latent positions $\bm{X}$.

Consider first the case without observed covariates. Let $\bm{P}$ be positive semidefinite and let $h$ be the identity function $h(u) = u$. In this setting, we only have latent positions $\bm{X}$, that are unobserved. If we were able to observe $\bm{P}= \bm{X}\bm{X}^T $, estimation of $\bm{X}$ would be straightforward. Furthermore, we could use spectral embeddings for $\bm{P}$ by exploiting the fact that $\bm{P}$ is positive semidefinite of rank $d$ and has spectral decomposition $\bm{P} = \bm{U}_{\bm{P}} \bm{S}_{\bm{P}} \bm{U}_{\bm{P}}^{T}$, where $\bm{S}_{\bm{P}}$ is a diagonal matrix containing the largest $d$ eigenvalues (in absolute value) of $\bm{P}$ and $\bm{U}_{\bm{P}}$ is the matrix with the corresponding eigenvectors. This implies that a good estimate for $\bm{X}$ is $\widehat{\bm{X}} = \bm{U}_{\bm{P}} \vert \bm{S}_{\bm{P}} \vert^{1/2} $, where $\vert \cdot \vert$ denotes entrywise absolute values. The estimation problem arises because \emph{we only observe $\bm{A}$, a perturbed version of $\bm{P}$}. The Adjacency Spectral Embedding of $\bm{A}$ into $\mathbb{R}^d$ is then $\widehat{\bm{X}} = \bm{U}_{\bm{A}} \vert \bm{S}_{\bm{A}}\vert^{1/2}$
where $\bm{S}_{\bm{A}}$ is a diagonal matrix containing the largest $d$ eigenvalues of $\bm{A}$ in absolute value and $\bm{U}_{\bm{A}}$ is the matrix with the corresponding eigenvectors.

In our asymptotic results for the above setup, we use the fact that ASE estimates of latent positions $\bm{X}$ asymptotically achieve perfect clustering (moreover, are asymptotically normal) and can be identified up to multiplication by an orthogonal matrix \citep{AthreyaEtAl2018a,TangEtAl2017}. This implies that asymptotically the blocks are recovered exactly (up to relabeling). The same logic and results hold for non-positive definite matrices $\bm{P}$, allowing us to study more general stochastic blockmodels \citep{Rubin-DelanchyEtAl2018}.

\subsection{Overview of the method in a 2-block SBM with one binary covariate}

To illustrate the methodology and to develop intuition, we focus on the special case of a $K=2$ stochastic blockmodel with a single discrete covariate and with latent positions in the unit interval $[0,1]$, yielding $d_{1}=1, d_{2}=0,$ and $r=1$, where $Z_i\in \lbrace 0,1\rbrace$ is a binary variable
(e.g.,~male/female, white/nonwhite, rich/poor, etc.) and the function $f(Z_i,Z_j;\beta)= \beta \mathbf{1}_{\lbrace Z_i = Z_j \rbrace}$ is an indicator for the equality of the covariates for $i$ and $j$, weighted by the parameter $\beta$. The main advantage of this approach is that we can illustrate the geometry of the method in a low-dimensional space. In our simple example, the matrix $\bm{B}$ is given by
\begin{equation}
\bm{B} = \bordermatrix{ & block_1 & block_2 \cr
                block_1& p^2  & pq \cr
                block_2& pq & q^2  },
\label{eq:matrix_B}
\end{equation}
where $p,q\in[0,1]$. We can conveniently re-write the matrix $\bm{B}$ as a dot-product of vector $\bm{\nu} = [p \ \ q]^T$, with $p,q\in[0,1]$, that is $\bm{B}=\bm{\nu}\bm{\nu}^T$, 
so that the SBM can be re-written as a random dot-product graph model with $\bm{X}_i = p$ if $i$ is in block 1, $\bm{X}_i = q$ if $i$ is in block 2. The probability of linking can then be written as
\begin{equation}
\bm{P}_{ij} = h\left(\bm{X}_i^T \bm{X}_j + \beta \mathbf{1}_{\lbrace Z_i = Z_j \rbrace}\right).
\label{eq:general_model}
\end{equation}
For ease of exposition the network blocks have the same probability, 
so $(\pi_1, \pi_2) = (0.5,0.5)$ and each community contains half males ($Z_i=1$) and half females ($Z_i=0$). However, we note that our algorithm and the theoretical results are valid when we allow the blocks to be of different size, and the observed covariates to be correlated with the unobserved blocks.

The model specified via (\ref{eq:general_model}) corresponds to a 4-block stochastic blockmodel. Indeed, we have 2 unobserved blocks, that are split in two additional blocks by the observed binary variable. Therefore, the final result is a 4-block SBM. More generally, if there are $K$ latent blocks and one binary covariates, we will have a $\tilde{K}=2K$-block SBM.

The possible values of $\bm{X}_i^T \bm{X}_j$ are $\lbrace p^2, pq, q^2 \rbrace$. Therefore the 4-block model can be completely characterized by the $4\times 4$ matrix 
\begin{equation}
\bm{B}_Z = \bordermatrix{ & male_1 &female_1 & male_2 & female_2 \cr
                male_1& p^2 + \beta & p^2  & pq + \beta & pq \cr
                female_1& p^2  & p^2 + \beta & pq & pq + \beta \cr
                male_2& pq + \beta & pq & q^2 + \beta & q^2 \cr
                female_2& pq  &  pq + \beta & q^2 & q^2 + \beta }.
\label{eq:matrix_BZ_shape}
\end{equation}
The value $h(\bm{B}_{Z,11}) = h(p^2 + \beta)$ is the probability that two males in block 1 form a link; on the other hand, $h(\bm{B}_{Z,12})=h(p^2)$ is the probability that a male and a female in block 1 form a link; $h(\bm{B}_{Z,31})=h(pq+\beta)$ is the probability that two males, one in block 1 and one in block 2, form a link; and so on. 

The above observations imply that, for this four block SBM, there exists a corresponding GRDPG with link probability matrix
\begin{equation}
\bm{P}  = \bm{Y} \bm{I}_{d_1,d_2} \bm{Y}^T,
\end{equation}
for some $n\times d$ matrix of latent positions $\bm{Y}$ with $d_{1} \ge 1$, $d_{2} \ge 0$, and $d=d_1+d_2$. 

To estimate the parameter $\beta$ and the latent positions $p$ and $q$ we use the following algorithmic approach. 

\begin{enumerate}
    \item We compute an eigendecomposition of the adjacency matrix $\bm{A}$, letting $\bm{S}_A$ denote the matrix whose diagonal contains the largest $\hat{d}$ eigenvalues of $\bm{A}$ in absolute value and $\bm{U}_A$ denote the matrix whose columns are corresponding unit norm eigenvectors.\footnote{
    In principle, the optimal value for $\hat{d} = rank(\bm{B}_Z)$; however we do not observe $\bm{B}_Z$, therefore we estimate $\hat{d}$ by profile likelihood methods \citep{ZhuGhodsi2006}.} The Adjacency Spectral Embedding (ASE) of $\bm{A}$ gives an estimate of the latent positions of the 4-block model as 
    \begin{equation}
        \widehat{\bm{Y}} = \bm{U}_A \vert \bm{S}_{A}\vert^{1/2},
    \end{equation}
    where $\vert \cdot \vert$ indicates the absolute value (entrywise).
    \item We use $\hat{\bm{Y}}$ to estimate $\bm{P}$ as
    \begin{equation}
        \widehat{\bm{P}} = \widehat{\bm{Y}} \bm{I}_{\hat{d}_1,\hat{d}_2}\widehat{\bm{Y}}^T,
    \end{equation}
    where 
    $\hat{d} := \hat{d}_{1} + \hat{d}_{2}$ is the number of largest eigenvalues (in magnitude) of $\bm{A}$ beyond a prescribed threshold, and $\hat{d}_1$ and $\hat{d}_2$ are the number of these eigenvalues that are positive and negative, respectively.

    \item We use a clustering procedure to assign each row of $\widehat{\bm{Y}}$ to one of $\tilde{K}=4$ blocks. We use a Gaussian Mixture Modeling approach (GMM) and estimate the center of the clusters  $\widehat{\bm{\mu}}=[\widehat{\bm{\mu}}_1,\widehat{\bm{\mu}}_2,\widehat{\bm{\mu}}_3,\widehat{\bm{\mu}}_4]$, that is the means of the Gaussians from the GMM.
    \item We compute an estimate for $\bm{\theta}_Z$  as
    \begin{eqnarray}
        \widehat{\bm{\theta}}_Z &=& \widehat{\bm{\mu}}^{T}\bm{I}_{\hat{d}_1,\hat{d}_2}\widehat{\bm{\mu}} \notag\\
        &=&  \bordermatrix{ & male_1 &female_1 & male_2 & female_2 \cr
                male_1& \widehat{\bm{\mu}}_1^T\bm{I}_{\hat{d}_1,\hat{d}_2}\widehat{\bm{\mu}}_1  & \widehat{\bm{\mu}}_1^T\bm{I}_{\hat{d}_1,\hat{d}_2}\widehat{\bm{\mu}}_2 & \widehat{\bm{\mu}}_1^T\bm{I}_{\hat{d}_1,\hat{d}_2}\widehat{\bm{\mu}}_3 & \widehat{\bm{\mu}}_1^T\bm{I}_{\hat{d}_1,\hat{d}_2}\widehat{\bm{\mu}}_4 \cr
                female_1& \widehat{\bm{\mu}}_2^T\bm{I}_{\hat{d}_1,\hat{d}_2}\widehat{\bm{\mu}}_1  & \widehat{\bm{\mu}}_2^T\bm{I}_{\hat{d}_1,\hat{d}_2}\widehat{\bm{\mu}}_2 & \widehat{\bm{\mu}}_2^T\bm{I}_{\hat{d}_1,\hat{d}_2}\widehat{\bm{\mu}}_3 &\widehat{\bm{\mu}}_2^T\bm{I}_{\hat{d}_1,\hat{d}_2}\widehat{\bm{\mu}}_4 \cr
                male_2& \widehat{\bm{\mu}}_3^T\bm{I}_{\hat{d}_1,\hat{d}_2}\widehat{\bm{\mu}}_1 & \widehat{\bm{\mu}}_3^T\bm{I}_{\hat{d}_1,\hat{d}_2}\widehat{\bm{\mu}}_2 & \widehat{\bm{\mu}}_3^T\bm{I}_{\hat{d}_1,\hat{d}_2}\widehat{\bm{\mu}}_3 & \widehat{\bm{\mu}}_3^T\bm{I}_{\hat{d}_1,\hat{d}_2}\widehat{\bm{\mu}}_4 \cr
                female_2& \widehat{\bm{\mu}}_4^T\bm{I}_{\hat{d}_1,\hat{d}_2}\widehat{\bm{\mu}}_1  &  \widehat{\bm{\mu}}_4^T\bm{I}_{\hat{d}_1,\hat{d}_2}\widehat{\bm{\mu}}_2 & \widehat{\bm{\mu}}_4^T\bm{I}_{\hat{d}_1,\hat{d}_2}\widehat{\bm{\mu}}_3 & \widehat{\bm{\mu}}_4^T\bm{I}_{\hat{d}_1,\hat{d}_2}\widehat{\bm{\mu}}_4 }.
                \label{eq:estimated_matrix_BZ}
    \end{eqnarray}
    By comparing the matrix $\widehat{\bm{\theta}}_Z $ and the population matrix $\bm{\theta}_Z $, we can assign each of the 4 blocks to the original 2 blocks. In fact, we know that the diagonal terms of $\widehat{\bm{\theta}}_Z $ are estimates of $h(p^2+\beta)$ if the nodes belong to block 1, or $h(q^2 + \beta)$ if nodes belong to block 2. This observation shows that we can group the 4 entries of the diagonal, by checking which values are close. In our case, we will group  $\widehat{\bm{\mu}}_1^T\bm{I}_{\hat{d}_1,\hat{d}_2}\widehat{\bm{\mu}}_1 $ and $\widehat{\bm{\mu}}_2^T\bm{I}_{\hat{d}_1,\hat{d}_2}\widehat{\bm{\mu}}_2 $ in one block, and $\widehat{\bm{\mu}}_3^T\bm{I}_{\hat{d}_1,\hat{d}_2}\widehat{\bm{\mu}}_3$ and $\widehat{\bm{\mu}}_4^T\bm{I}_{\hat{d}_1,\hat{d}_2}\widehat{\bm{\mu}}_4$ in another block. Therefore, blocks 1 and 2 in the 4-block model are assigned to the original latent block 1, while blocks 3 and 4 are assigned to original latent block 2. 
    \item We estimate $\beta$ from the entries of the matrix $ \widehat{\bm{B}}_Z  = h^{-1}(\widehat{\bm{\theta}_Z })$, where the inverse function $h^{-1}$ is applied element-wise. For example we know that  $h^{-1}(\widehat{\bm{\mu}}_1^T\bm{I}_{\hat{d}_1,\hat{d}_2}\widehat{\bm{\mu}}_1) $ is an estimate of $p^2+\beta$ or $q^2 +\beta$ (because of the invariance of the model to relabeling of the blocks). Without loss of generality, assume that $h^{-1}(\widehat{\bm{\mu}}_1^T\bm{I}_{\hat{d}_1,\hat{d}_2}\widehat{\bm{\mu}}_1) $ is an estimate of $p^2+\beta$; therefore, the entry  $h^{-1}(\widehat{\bm{\mu}}_1^T\bm{I}_{\hat{d}_1,\hat{d}_2}\widehat{\bm{\mu}}_2 )$ is an estimate of $p^2$. Therefore, the point estimate of $\beta$ is then
    \begin{equation}
        \widehat{\beta} = h^{-1}( \widehat{\bm{\mu}}_1^T\bm{I}_{\hat{d}_1,\hat{d}_2}\widehat{\bm{\mu}}_1) -h^{-1}(\widehat{\bm{\mu}}_1^T\bm{I}_{\hat{d}_1,\hat{d}_2}\widehat{\bm{\mu}}_2).
    \end{equation}
    \item The latent positions $p$ and $q$ can be estimated from the matrix  $\widehat{\bm{B}}_Z $, by using the submatrix
        \begin{equation}
          \bordermatrix{  &female_1  & female_2 \cr
                male_1&  h^{-1}(\widehat{\bm{\mu}}_1^T\bm{I}_{\hat{d}_1,\hat{d}_2}\widehat{\bm{\mu}}_2) &  h^{-1}(\widehat{\bm{\mu}}_1^T\bm{I}_{\hat{d}_1,\hat{d}_2}\widehat{\bm{\mu}}_4) \cr
                male_2&  h^{-1}(\widehat{\bm{\mu}}_3^T\bm{I}_{\hat{d}_1,\hat{d}_2}\widehat{\bm{\mu}}_2) &  h^{-1}(\widehat{\bm{\mu}}_3^T\bm{I}_{\hat{d}_1,\hat{d}_2}\widehat{\bm{\mu}}_4)  } = 
            \bordermatrix{  &   &   \cr
                 &  \widehat{p}^2 &\widehat{p}\widehat{q} \cr 
               &  \widehat{p}\widehat{q} &  \widehat{q}^2 }.
\label{eq:estimated_matrix_BZ_latentpos}
    \end{equation}
The spectral embedding of this matrix provides estimates for the latent positions $\hat{p}$ and $\hat{q}$, that are identified up to an orthogonal transformation.

\end{enumerate}

In practice, we can estimate $\beta$ from multiple entries of the matrix $\bm{B}_Z$, for example $\beta = \bm{B}_{Z,11}-\bm{B}_{Z,12} = \bm{B}_{Z,33}-\bm{B}_{Z,34} $, and weight each estimate by the size of the blocks. This could improve the estimate, since some blocks are larger than others, so delivering more precise estimates. Our code implements this idea, which is more practical for empirical applications.

\section{Asymptotic theory}
In this section, we derive a central limit theorem for the spectral estimator of $\beta$. For ease of exposition, we focus on the case of a single binary observed covariate and scalar $\beta$, though our method works for other specifications in which the effect of the observed covariates $\beta$ can be written as a function of the stochastic blockmodel's probability matrix $\bm{\theta}_Z$. Extensions to multiple binary or discrete observed covariates are straightforward albeit tedious.

We desire to estimate a stochastic blockmodel with observed covariates, where
\begin{eqnarray}
\tau_i &\overset{iid}{\sim} & Multinomial(1;\pi_1,\dots,\pi_K), \\
\bm{Z}_i\vert \tau_i &\overset{ind}{\sim} & Bernoulli(b_{\tau_i}),\\
\bm{A}_{ij}\vert \tau_i, \tau_j, \bm{Z}_i, \bm{Z}_j &\overset{ind}{\sim}& Bernoulli(\bm{P}_{ij}),\\
\bm{P}_{ij}  &=& h\left(\bm{B}_{\tau_i \tau_j} + \beta \bm{1}_{\lbrace  \bm{Z}_i = \bm{Z}_j\rbrace}\right). \label{eq:general_model_sbm}
\end{eqnarray}

We assume that the observed covariates are binary and can depend on the block assignment, that is $\bm{Z}_i\vert \tau_i \overset{ind}{\sim} Bernoulli(b_{\tau_i})$, where $b_{\tau_i}= P(\bm{Z}_i=1\vert \tau_i)$ . Our asymptotic results are easily extended to the case of discrete observed covariates with three or more possible outcomes. 

As explained above in the simple example, our strategy consists of rewriting the SBM as a GRDPG. First, notice that the matrix $\bm{B}$ can be written as $\bm{B}_{\tau_i \tau_j} = \bm{X}_i^T \bm{X}_j$, where $\bm{X}_i$ is a $d \times 1$ 
vector of latent positions that has $K$ possible values  $\bm{\nu}_1, \dots, \bm{\nu}_{K}$. In practice, $\bm{\nu}=(\bm{\nu}_1, \dots, \bm{\nu}_{K})$  are the centers of the $K$ blocks $\bm{X}$, such that $i$ and $j$ belong to \emph{unobserved} block $k$ when $\bm{X}_i=\bm{X}_j=\bm{\nu}_k$. Let $\bm{\tau}$ be the function that assigns nodes to unobserved blocks; then $\tau_i=k$ if $\bm{X}_i=\bm{\nu}_k$. We can thus rewrite the stochastic blockmodel above as a random dot product graph with observed covariates as follows:
\begin{eqnarray}
\bm{X}_i &\overset{iid}{\sim} & \pi_1 \delta_{\nu_1}+ \pi_2 \delta_{\nu_2}+\cdots+\pi_K \delta_{\nu_K}, \\
\bm{Z}_i\vert \bm{X}_i &\overset{ind}{\sim} & Bernoulli(b_{\tau_i}),\\
\bm{A}_{ij}\vert \bm{X}_i, \bm{X}_j, \bm{Z}_i, \bm{Z}_j &\overset{ind}{\sim}& Bernoulli(\bm{P}_{ij}),\\
\bm{P}_{ij}  &=& h\left(\bm{X}_i^T \bm{X}_j + \beta \bm{1}_{\lbrace  \bm{Z}_i = \bm{Z}_j\rbrace}\right).
\end{eqnarray}

We first notice that both models are stochastic blockmodels with $\tilde{K}=2K$ blocks, because the indicator variable $\bm{1}_{\lbrace \bm{Z}_i=\bm{Z}_j\rbrace}$ splits each unobserved block in two blocks. 
The probabilities of belonging to a block $k$ for this $\widetilde{K}$-block SBM are denoted as  $\bm{\eta} = (\eta_1, \dots, \eta_{\widetilde{K}}) = ( \pi_1 \cdot b_1, \pi_1 \cdot (1-b_1), \pi_2 \cdot b_2, \pi_2 \cdot (1-b_2), \dots , \pi_K \cdot b_K, \pi_K \cdot (1-b_K) )$; and the functions that assign nodes to blocks are $\bm{\xi} = (\xi_1,\dots,\xi_{n} )$, such that  $\xi_i = 1$ if $ \tau_i = 1$ and $ \bm{Z}_i=0 $; $\xi_i =2$ if $\tau_i=1$ and $ \bm{Z}_i=1$; $\xi_i=3$ if $\tau_i=2$ and $ \bm{Z}_i=0 $; $\xi_i=4$ if $\tau_i=2$ and $ \bm{Z}_i=1$; and so on.\\

So we have a stochastic blockmodel $ (\bm{A}, \bm{\xi}, \bm{Z}) \sim SBM(\bm{\theta}_Z,\bm{\eta})$ with the $\widetilde{K}\times \widetilde{K}$ matrix of probabilities $\bm{\theta}_Z$ 
\begin{tiny}
\begin{equation}
\bm{\theta}_Z = \bordermatrix{ & \tau=1;Z=0 & \tau=1;Z=1 & \tau=2;Z=0 & \tau=2;Z=1 & \cdots & \tau=K;Z=0 & \tau=K;Z=1 \cr
                \tau=1;Z=0 & h\left(\bm{\nu}_1^T  \bm{\nu}_1 + \beta \right) & h\left(\bm{\nu}_1^T  \bm{\nu}_1\right)  & h\left(\bm{\nu}_1^T  \bm{\nu}_2  + \beta \right) & h\left(\bm{\nu}_1^T \bm{\nu}_2 \right) & \cdots & h\left(\bm{\nu}_1^T  \bm{\nu}_K  + \beta \right) & h\left(\bm{\nu}_1^T  \bm{\nu}_K \right)\cr                
                \tau=1;Z=1 & h\left(\bm{\nu}_1^T  \bm{\nu}_1\right) & h\left(\bm{\nu}_1^T  \bm{\nu}_1 + \beta\right) & h\left(\bm{\nu}_1^T \bm{\nu}_2 \right) & h\left( \bm{\nu}_1^T  \bm{\nu}_2 + \beta\right) & \cdots & h\left(\bm{\nu}_1^T  \bm{\nu}_K \right)  & h\left(\bm{\nu}_1^T  \bm{\nu}_K  + \beta \right)\cr
                \tau=2;Z=0 & h\left(\bm{\nu}_2^T  \bm{\nu}_1 + \beta \right)& h\left(\bm{\nu}_2^T  \bm{\nu}_1 \right) & h\left(\bm{\nu}_2^T \bm{\nu}_2  + \beta\right) & h\left(\bm{\nu}_2^T  \bm{\nu}_2 \right)& \cdots & h\left(\bm{\nu}_2^T  \bm{\nu}_K  + \beta \right)& h\left(\bm{\nu}_2^T \bm{\nu}_K \right)\cr                
                \tau=2;Z=1 & h\left(\bm{\nu}_2^T  \bm{\nu}_1\right) & h\left(\bm{\nu}_2^T \bm{\nu}_1 + \beta \right)& h\left(\bm{\nu}_2^T  \bm{\nu}_2\right) & h\left(\bm{\nu}_2^T \bm{\nu}_2 + \beta \right) & \cdots & h\left(\bm{\nu}_2^T \bm{\nu}_K\right)  & h\left(\bm{\nu}_2^T \bm{\nu}_K  + \beta \right)\cr
                \vdots &  \vdots & \vdots & & & \ddots \cr
                \tau=K;Z=0 & h\left(\bm{\nu}_K^T  \bm{\nu}_1 + \beta\right) & h\left(\bm{\nu}_K^T  \bm{\nu}_1 \right) & h\left(\bm{\nu}_K^T \bm{\nu}_2  + \beta \right) & h\left(\bm{\nu}_K^T \bm{\nu}_2 \right) & \cdots & h\left( \bm{\nu}_K^T \bm{\nu}_K  + \beta \right) & h\left(\bm{\nu}_K^T  \bm{\nu}_K \right) \cr                
                \tau=K;Z=1 & h\left( \bm{\nu}_K^T  \bm{\nu}_1 \right) & h\left( \bm{\nu}_K^T  \bm{\nu}_1 + \beta \right) & h\left( \bm{\nu}_K^T  \bm{\nu}_2 \right) & h\left( \bm{\nu}_K^T  \bm{\nu}_2 + \beta \right) & \cdots & h\left( \bm{\nu}_K^T  \bm{\nu}_K \right)  & h\left( \bm{\nu}_K^T   \bm{\nu}_K  + \beta\right)  \cr
 }
\label{eq:matrix_thetaZ_shape}
\end{equation}
\end{tiny}

The stochastic blockmodel characterized by the matrix $\bm{\theta}_Z$ can be re-formulated as a GRDPG. Indeed, consider the eigendecomposition $\bm{\theta}_Z \equiv \bm{U}\bm{\Sigma} \bm{U}^T$, and define $\bm{\mu}=[\bm{\mu}_1,\bm{\mu}_2, \dots,\bm{\mu}_{\widetilde{K}}]$ as the rows of $\bm{U}\vert \bm{\Sigma}\vert^{1/2}$; then let $F=\sum_{k=1}^{\widetilde{K}}\eta_k \delta_{\bm{\mu}_k}$, where $\delta$ is the Dirac-delta; and $d_1$ and $d_2$ are the number of positive and negative eigenvalues of $\bm{\theta}_Z$, respectively. Then, the Generalized Random Dot Product Graph model  $(\bm{Y},\bm{A}) \sim GRDPG_{d_1,d_2}(F)$ corresponding to our stochastic blockmodel $ (\bm{A}, \bm{\xi}, \bm{Z}) \sim SBM(\bm{\theta}_Z,\bm{\eta})$ is given by
\begin{eqnarray}
    \bm{Y}_i &\overset{iid}{\sim} & \eta_1 \delta_{\bm{\mu}_1} + \cdots + \eta_{\tilde{K}} \delta_{\bm{\mu}_{\tilde{K}}}, \\
    \bm{A}_{ij} \vert \bm{Y}_i,\bm{Y}_j &\overset{ind}{\sim} & Bernoulli(\bm{Y}_i^T \bm{I}_{d_1,d_2} \bm{Y}_j),
\end{eqnarray}
where  $d_1+d_2=\widetilde{d}=rank(\bm{\theta}_Z)$ and $\bm{Y}$ is the $n\times \widetilde{d}$  vector of latent positions with centroids $\bm{\mu}$.\\

We can now extend asymptotic results for estimation of RDPGs in \cite{AthreyaEtAl2018a,TangEtAl2017} to estimate block assignments and the effect of the covariates (see \cite{Rubin-DelanchyEtAl2018} for the corresponding generalization to GRDPGs). \\

\subsection{Main theoretical result}
Because the functions $\bm{\tau}$ that describe the assignments to blocks are unknown, the $\widetilde{K}$ SBM model assignment functions $\bm{\xi}$ are also unknown. Applying  the Adjacency Spectral Embedding procedure, we recover an estimate $\widehat{\bm{\xi}}$. 

We prove asymptotic normality for the parameter $\beta$, exploiting the fact that $\beta$ can be written as a function of the SBM probabilities, that is 
\begin{equation}
    \beta = h^{-1}\left(\bm{\theta}_{Z,11}\right) - h^{-1}\left(\bm{\theta}_{Z,12}\right) = h^{-1}\left(\bm{\nu}_1^T \bm{\nu}_1 + \beta \right)-  h^{-1}\left(\bm{\nu}_1^T \bm{\nu}_1 \right).
\end{equation}

If the blocks were known at the onset, we could use the estimator $\widehat{\beta} = h^{-1}(\widehat{\bm{\theta}}_{Z,11} )- h^{-1}(\widehat{\bm{\theta}}_{Z,12}  )$. However, all that we have access to is the estimate $\widehat{\bm{\xi}}$, so it is crucial that this estimate be consistent. For RDPGs this is indeed the case, as one can prove that the latent blocks are recovered up to an orthogonal transformation matrix in the large $n$ limit (Lemma~4 in \cite{TangEtAl2017}). Therefore we can recover the parameter $\beta$ up to relabeling of the blocks. This is summarized in the following theorem.

\begin{theorem} \textbf{Central limit theorem for $\beta$} \\
Let $\bm{\tau}$ be unknown and $K$ known. Let $\widehat{\bm{\tau}}: [n]\rightarrow [K]$ be the function that assigns nodes to clusters, estimated using GMM or K-means clustering on the rows of $\widehat{\bm{Y}}= \widehat{\bm{U}}\vert \widehat{\bm{S}} \vert^{1/2}$ .  Let function $g$ be defined as the inverse of $h$, that is $g(\cdot)=h^{-1}(\cdot)$, with first derivative $g^{\prime}(\cdot)$. Let $g^{\prime}(\bm{\nu}_{1}^T \bm{\nu}_{1} +\beta)\neq 0$ and $g^{\prime}(\bm{\nu}_{1}^T \bm{\nu}_{2} )\neq 0$. Then there exists a sequence of permutations $\phi\equiv \phi_n$ on $[K]$ such that the estimator  $\widehat{\beta} = h^{-1}(\widehat{\bm{\theta}}_{Z,\phi(1)\phi(1)})-h^{-1}(\widehat{\bm{\theta}}_{Z,\phi(1)\phi(2)}) $ is asymptotically normal, that is
\begin{equation}
    n\left(\widehat{\beta} - \beta - \frac{\widehat{\psi}_\beta}{n}\right)\overset{d}{\rightarrow}N(0,\widehat{\sigma}_\beta^2)
\end{equation}
as $n\rightarrow\infty$. The values $\widehat{\psi}_\beta$ and $\widehat{\sigma}_\beta^2$ are computed in the proof.
\label{thm:clt_general_h_block_estim}
\end{theorem}
\begin{proof}
See Appendix. 
\end{proof}


\subsection{Sparsity}
The previous theoretical result implicitly assumes a dense network. However, many social and economic networks of interest in applications display some degree of sparsity. This is an empirical regularity 
that social scientists have observed in many settings, as most people do not form many links. Economists rationalize sparsity with the fact that people have constraints on time to spend with their friends \citep{Jackson2008}. 

We follow the literature and assume that sparsity is an asymptotic feature of the data generating process.  We multiply the probability $\bm{P}_{ij}$ by a scalar $\rho_n$ that governs the sparsity of the network, that is, the probability of a link between nodes $i$ and $j$ becomes
\begin{equation}
    \bm{P}_{ij} = \rho_n h\left( \bm{X}_i^T \bm{X}_j + \beta \mathbf{1}_{\lbrace\bm{Z}_i = \bm{Z}_j\rbrace}\right).
\end{equation}

Our previous result in Theorem \ref{thm:clt_general_h_block_estim} applies to dense networks; that is when $\rho_n\rightarrow c$ where $c\in (0,1]$ is a constant. For simplicity and without loss of generality, in Theorem \ref{thm:clt_general_h_block_estim} we have assumed $c=1$. 

In this section we consider formally the case of $\rho_n\rightarrow 0$ as $n\rightarrow\infty$. We have to limit the rate of convergence for $\rho_n$, because the network could become too sparse, not allowing estimation. We will describe this regime a \emph{semi-sparse}, because we will allow $\rho_n\rightarrow 0$ but $n\rho_n = \omega(\sqrt{n})$, that is the average degree of the network grows sub-linearly in $n$.\footnote{The notation $n\rho_n = \omega(\sqrt{n})$ means that for any real constant $a>0$ there exists an $n_0\geq 1$ such that $\rho_n > a /\sqrt{n}\geq 0$ for every integer $n\geq n_0$.}
The intuition for this restriction is that too much sparsity makes links ``too rare'' and therefore spectral estimation and inference are impeded by having too few observations. 

\begin{theorem} \textbf{Central limit theorem for sparse networks}\\
Let model (\ref{eq:general_model_sbm}) include a sparsity coefficient $\rho_n$
\begin{equation}
    \bm{P}_{ij} = \rho_n h\left( \bm{X}_i^T \bm{X}_j + \beta \mathbf{1}_{\lbrace\bm{Z}_i = \bm{Z}_j\rbrace}\right)
\end{equation}
such that $\rho_n\rightarrow 0$ and $n\rho_n = \omega(\sqrt{n})$ as $n\rightarrow\infty$. Let $\hat{\bm{\tau}}$  be assignment of each node to a block, estimated using ASE and GMM (or K-means) clustering. Then there exists a sequence of permutations $\phi\equiv \phi_n$ on $[K]$ such that the estimator  $\widehat{\beta} = h^{-1}(\widehat{\bm{\theta}}_{Z,\phi(1)\phi(1)})-h^{-1}(\widehat{\bm{\theta}}_{Z,\phi(1)\phi(2)}) $ is asymptotically normal, that is
\begin{eqnarray}
n \rho_n^{1/2} \left(\hat{\beta} - \beta - \frac{\ddot{\psi}_{\beta}}{n\rho_n} \right)\overset{d}{\longrightarrow} N\left(0,\ddot{\sigma}_{\beta}^2 \right)
\end{eqnarray}
where $\ddot{\psi}_\beta$ and $\ddot{\sigma}_\beta^2$  are computed in the proof.

\label{thm:clt_general_h_semisparse}
\end{theorem}
\begin{proof}
See Appendix. 
\end{proof}
 
 Theorem~\ref{thm:clt_general_h_semisparse} says that as long as the network is not too sparse, the estimator of $\beta$ will be asymptotically normal. Notably, the bias term does not vanish asymptotically.
 
 We note that our formulation of sparsity does not impose any restriction on the network data for a fixed $n$, as it is based on a large sample property.

\subsection{Multiple observed covariates}
The asymptotic results hold for discrete observed covariates and more general models, as long as the effect of the observed covariates on the probability of linking can be written as a function of the block probabilities. Let the observed variables $\bm{Z}_i = [\bm{Z}_{i}^{(1)},\bm{Z}_{i}^{(2)}]$ be two covariates. For simplicity, we consider the case of binary variables, and we assume  $\bm{Z}_{i}^{(1)}\overset{ind}{\sim} Bernoulli(b_{\tau_i}^{(1)})$ and $\bm{Z}_{i}^{(2)}\overset{ind}{\sim} Bernoulli(b_{\tau_i}^{(2)})$. The results still hold for discrete variables. The model is 
\begin{eqnarray}
\bm{A}_{ij} \vert \bm{X}_i,\bm{X}_j,\bm{Z}_i,\bm{Z}_j &\overset{ind}{\sim} & Bernoulli\left(\bm{P}_{ij}\right), \\
\bm{P}_{ij} &=& h\left(  \bm{X}_i^T \bm{X}_j + \beta_1 \bm{1}_{\lbrace\bm{Z}_{i}^{(1)} = \bm{Z}_{j}^{(1)}\rbrace} + \beta_2 \bm{1}_{\lbrace\bm{Z}_{i}^{(2)} = \bm{Z}_{j}^{(2)}\rbrace}\right).
\end{eqnarray}
This stochastic blockmodel has $\tilde{K}=4K$ blocks, $ (\bm{A}, \bm{\xi}, \bm{Z}) \sim SBM(\bm{\theta}_{Z},\bm{\eta})$ with $\widetilde{K}\times \widetilde{K}$ matrix of probabilities $\bm{\theta}_{Z}$  given by

\begin{eqnarray}
\bm{\theta}_Z = 
\left[ \begin{matrix}
\bm{W}_{11}  & \bm{W}_{12} & ... & \bm{W}_{1\tilde{K}}\\
\bm{W}_{21} &  \bm{W}_{22} & ... & \bm{W}_{2\tilde{K}}\\
\vdots  &  \vdots  & \ddots & \vdots\\
\bm{W}_{\tilde{K}1}  & \bm{W}_{\tilde{K}2} & ... & \bm{W}_{\tilde{K}\tilde{K}}\\

\end{matrix} \right]
\label{eq:matrix_thetaZS_shape}
\end{eqnarray}

where each matrix $\bm{W}_{k\ell}$ is given by 

\begin{eqnarray}
\begin{blockarray}{c|cccc}
 & \multirow{3}{2cm}{ $\tau=\ell$   $Z^{(1)}=0$ $Z^{(2)}=0$ } &  \multirow{3}{2cm}{ $\tau=\ell$   $Z^{(1)}=1$ $Z^{(2)}=0$ } & \multirow{3}{2cm}{ $\tau=\ell$   $Z^{(1)}=0$ $Z^{(2)}=1$ }  & \multirow{3}{2cm}{ $\tau=\ell$   $Z^{(1)}=1$ $Z^{(2)}=1$}    
  \\
  \\ 
  \\
  \hline
 \multirow{3}{2cm}{ $\tau=k$   $Z^{(1)}=0$ $Z^{(2)}=0$ } & \multirow{3}{3.5cm}{ $h(\bm{\nu}_k^T \bm{\nu}_\ell+ \beta_1 + \beta_2)$ } & \multirow{3}{3cm}{ $h(\bm{\nu}_k^T \bm{\nu}_\ell+ \beta_2)$ } & \multirow{3}{3cm}{ $h(\bm{\nu}_k^T \bm{\nu}_\ell+ \beta_1) $ } & \multirow{3}{1.5cm}{ $h(\bm{\nu}_k^T \bm{\nu}_\ell)$ }  
\\
\\
\\
\\
  \multirow{3}{2cm}{ $\tau=k$   $Z^{(1)}=1$ $Z^{(2)}=0$ } & \multirow{3}{3cm}{ $h(\bm{\nu}_k^T \bm{\nu}_\ell + \beta_2)$ } & \multirow{3}{3cm}{ $h(\bm{\nu}_k^T \bm{\nu}_\ell+\beta_1 + \beta_2)$ } & \multirow{3}{3cm}{ $h(\bm{\nu}_k^T \bm{\nu}_\ell) $ } & \multirow{3}{3cm}{ $h(\bm{\nu}_k^T \bm{\nu}_\ell + \beta_1)$ } \\
  \\
  \\
  \\
    \multirow{3}{2cm}{ $\tau=k$   $Z^{(1)}=0$ $Z^{(2)}=1$ } & \multirow{3}{3cm}{ $h(\bm{\nu}_k^T \bm{\nu}_\ell + \beta_1)$ } & \multirow{3}{3cm}{ $h(\bm{\nu}_k^T \bm{\nu}_\ell)$ } & \multirow{3}{3cm}{ $h(\bm{\nu}_k^T \bm{\nu}_\ell + \beta_1 + \beta_2)$ } & \multirow{3}{3cm}{ $h(\bm{\nu}_k^T \bm{\nu}_\ell + \beta_2)$ } \\ 
  \\
  \\
  \\   
 \multirow{3}{2cm}{ $\tau=k$   $Z^{(1)}=1$ $Z^{(2)}=1$ } & \multirow{3}{3cm}{ $h(\bm{\nu}_k^T \bm{\nu}_\ell) $ } & \multirow{3}{3cm}{ $h(\bm{\nu}_k^T \bm{\nu}_\ell + \beta_1)$ } & \multirow{3}{3cm}{ $h(\bm{\nu}_k^T \bm{\nu}_\ell  + \beta_2)$ } & \multirow{3}{3cm}{ $h(\bm{\nu}_k^T \bm{\nu}_\ell + \beta_1 + \beta_2)$ } \\ 
 \\
 \\
 \hline
\end{blockarray}
\end{eqnarray}

The intuition is the same as the model with one covariate. The blocks can be inferred by clustering the diagonal elements of matrix $\bm{\theta}_{Z}$, and the parameters $\beta_1$ and $\beta_2$ are functions of the $\bm{\theta}_{Z}$ entries, namely
\begin{eqnarray}
	\beta_1  = h^{-1}\left(\bm{\theta}_{Z,11}\right) - h^{-1}\left(\bm{\theta}_{Z,12}\right); \hspace{.5cm} 
	\beta_2  = h^{-1}\left(\bm{\theta}_{Z,11}\right) - h^{-1}\left(\bm{\theta}_{Z,13}\right). 
\end{eqnarray}
As such, the main characterization of the central limit theorem holds in this case with minimal modifications.

\subsection{Differential Homophily}
In many applications the researcher is interested in testing for differential homophily in observable characteristics. For example, homophily among males could be higher than homophily among females, other things being equal. 
This can be accomplished in our setting by a minor modification of the algorithm for estimation. 

For ease of exposition, we will again consider the model with $K=2$ blocks and a binary covariate $\bm{Z}_i\in \lbrace 0,1 \rbrace$, where for concreteness we assume that $\bm{Z}_i = 0$ denotes ``male''. The model with differential homophily is 
\begin{eqnarray}
\bm{A}_{ij} \vert \bm{X}_i,\bm{X}_j,\bm{Z}_i,\bm{Z}_j &\overset{ind}{\sim} & Bernoulli\left(\bm{P}_{ij}\right), \\
\bm{P}_{ij} &=& h\left(  \bm{X}_i^T \bm{X}_j + \beta_1 \bm{1}_{\lbrace\bm{Z}_{i} = \bm{Z}_{j} = 0\rbrace} + \beta_2 \bm{1}_{\lbrace\bm{Z}_{i} = \bm{Z}_{j}=1\rbrace}\right),
\end{eqnarray}
where $\beta_1$ measures the impact of being both males on the probability of a link, while $\beta_2$ measures the effect of homophily among females. 
We can thus write the corresponding probability matrix $\bm{\theta}_Z$ for this model as follows
\begin{eqnarray}
\begin{blockarray}{c|cccc}
 & \multirow{2}{2cm}{ $\tau=1$   $Z=0$  } &  \multirow{2}{2cm}{ $\tau=1$   $Z=1$ } & \multirow{2}{2cm}{ $\tau=2$   $Z=0$  }  & \multirow{2}{2cm}{ $\tau=2$   $Z=1$}    
  \\
  \\ 
  \hline
 \multirow{2}{2cm}{ $\tau=1$   $Z=0$ } & \multirow{2}{3cm}{ $h(\bm{\nu}_1^T \bm{\nu}_1+ \beta_1) $ } & \multirow{2}{3cm}{ $h(\bm{\nu}_1^T \bm{\nu}_1)$ } & \multirow{2}{3cm}{ $h(\bm{\nu}_1^T \bm{\nu}_2+ \beta_1) $ } & \multirow{2}{3cm}{ $h(\bm{\nu}_1^T \bm{\nu}_2)$ }  
\\
\\
\\
  \multirow{2}{2cm}{ $\tau=1$   $Z=1$ } & \multirow{2}{3cm}{ $h(\bm{\nu}_1^T \bm{\nu}_1)$ } & \multirow{2}{3cm}{ $h(\bm{\nu}_1^T \bm{\nu}_1 + \beta_2)$ } & \multirow{2}{3cm}{ $h(\bm{\nu}_1^T \bm{\nu}_2) $ } & \multirow{2}{3cm}{ $h(\bm{\nu}_1^T \bm{\nu}_2 + \beta_2)$ } \\
  \\
  \\
    \multirow{2}{2cm}{ $\tau=2$   $Z=0$ } & \multirow{2}{3cm}{ $h(\bm{\nu}_2^T \bm{\nu}_1 + \beta_1)$ } & \multirow{2}{3cm}{ $h(\bm{\nu}_2^T \bm{\nu}_1)$ } & \multirow{2}{3cm}{ $h(\bm{\nu}_2^T \bm{\nu}_2 + \beta_1 )$ } & \multirow{2}{3cm}{ $h(\bm{\nu}_2^T \bm{\nu}_2)$ } \\ 
  \\
  \\
 \multirow{2}{2cm}{ $\tau=2$   $Z=1$ } & \multirow{2}{3cm}{ $h(\bm{\nu}_2^T \bm{\nu}_1) $ } & \multirow{2}{3cm}{ $h(\bm{\nu}_2^T \bm{\nu}_1 + \beta_2)$ } & \multirow{2}{3cm}{ $h(\bm{\nu}_2^T \bm{\nu}_2 )$ } & \multirow{2}{3cm}{ $h(\bm{\nu}_2^T \bm{\nu}_2 + \beta_2)$ } \\ 
 \\
 \hline
\end{blockarray}
\end{eqnarray}

We notice that clustering the main diagonal cannot inform about the structure of the blocks as in the previous examples. However, if we were to know the block assignments we will be able to identify $\beta_1$ and $\beta_2$ a functions of the entries of  $\bm{\theta}_Z$
\begin{eqnarray}
\beta_1 = h^{-1}(\bm{\theta}_{Z,11} ) -  h^{-1}(\bm{\theta}_{Z,12} ); 
\hspace{.5cm} 
\beta_2 = h^{-1}(\bm{\theta}_{Z,22} ) -  h^{-1}(\bm{\theta}_{Z,12} )
\end{eqnarray}
Notice that both $\beta_1$ and $\beta_2$ are functions of the stochastic blockmodel probabilities, thus the main characterization of the central limit theorem still holds.

We can estimate the block structure in the following way. First, consider the subgraph among all nodes with $Z=0$. The corresponding probability matrix of this stochastic blockmodel is 
\begin{equation}
    \bm{\theta}_{00}  = \begin{bmatrix} 
    h(\bm{\nu}_1^T \bm{\nu}_1+ \beta_1) & \ \ \  h(\bm{\nu}_1^T \bm{\nu}_2+ \beta_1) \\
    & \\
    h(\bm{\nu}_2^T \bm{\nu}_1+ \beta_1) & \ \ \ h(\bm{\nu}_2^T \bm{\nu}_2+ \beta_1)
    \end{bmatrix}.
\end{equation}
The structure of the model ensures that this subgraph is a GRDPG and we can estimate the latent positions using our algorithm. We can then cluster the estimated latent positions to obtain the estimated block assignments. 

Second, we repeat the same procedure for the subgraph among nodes with $Z=1$. Third, we estimate the GRDPG using our algorithm, and use the estimated block assignments in the first and second step. This allows us to estimate both $\beta_1$ and $\beta_2$ through the formula above.

\subsection{Practical details about the estimators}
For ease of exposition let's consider the model with one binary covariate in Theorem \ref{thm:clt_general_h_block_estim}.\footnote{The estimator for the other cases (multiple covariates, differential homophily, discrete covariates, etc. is obtained analogously.} Our central limit theorem focus on the differences of two entries of the matrix $\bm{\theta}_Z$. However, we can compute $\beta$ in several ways, using different entries of the matrix, e.g.,~ 
$\beta = h^{-1}\left(\bm{\theta}_{Z,11}\right) - h^{-1}\left(\bm{\theta}_{Z,12}\right) = h^{-1}\left(\bm{\theta}_{Z,33}\right) - h^{-1}\left(\bm{\theta}_{Z,34}\right)$.
Therefore, we rely on two ways to estimate the model. The first estimator consists of computing all the values of $\beta$ from the relevant pairs of entries of $h^{-1}\left(\bm{\theta}_Z\right)$ and then averaging out. The second estimator weights each estimated $\beta$  by the size of the corresponding blocks.
Formally, for the first estimator, after estimating the block assignments $\widehat{\bm{\xi}}$, we compute the number of nodes in the block with a particular value of the covariate 
\begin{eqnarray}
n_{1,k}=\sum_{i:\widehat{\bm{\xi}}_i=k} \mathbf{1}_{\lbrace \bm{Z}_i= 1 \rbrace} \text{ \ \ and \ \ } n_{0,k}=\sum_{i:\widehat{\bm{\xi}}_i=k} \mathbf{1}_{\lbrace \bm{Z}_i= 0 \rbrace}
\end{eqnarray}
and we assign each block to a value of the covariate $\bm{Z}_{\theta,k} = 1$ if $n_{1,k}>n_{0,k}$ and $\bm{Z}_{\theta,k} = 0$ otherwise. Let $\bm{\psi}=(\psi_1,\cdots,\psi_{2K} ) \in \lbrace 1, \cdots, K\rbrace^{2K}$ be the vector that assigns each element of the diagonal of $\theta_Z$ to the corresponding unobserved block; let $\widehat{\bm{\psi}}$ be the corresponding estimated assignment. 
We then consider the set of all pairs set $M = \lbrace(k\ell, k\ell^\prime), k, \ell, \ell^\prime \in \lbrace 1, 2, \cdots, 2\widehat{K} \rbrace  \vert \widehat{\psi}_\ell = \widehat{\psi}_{\ell^{\prime}}, \bm{Z}_{\theta,k}=\bm{Z}_{\theta,\ell},  \bm{Z}_{\theta,k}\neq \bm{Z}_{\theta,\ell^{\prime}}\rbrace $. For each element $m = (k\ell, k\ell^\prime)\in M$ we compute the estimate $\widehat{\beta}_m$ as
\begin{eqnarray}
\widehat{\beta}_m = h^{-1}\left(\widehat{\bm{\theta}}_{Z,k\ell }\right) - h^{-1}\left(\widehat{\bm{\theta}}_{Z,k\ell^{\prime}}\right)
\end{eqnarray}
We then average out the values of the $\widehat{\beta}_m$ to obtain the final estimate
\begin{eqnarray}
\widehat{\beta}^{sa} = \frac{1}{\vert M \vert }\sum_{m \in M} \widehat{\beta}_m ,
\label{eq:simplemean_estimator}
\end{eqnarray}
where $\vert M \vert$ is the number of elements in set $M$, that is the number of paired entries in $h^{-1}\left(\widehat{\bm{\theta}}_Z\right)$ from which we can estimate $\beta$. 

The second estimator weights each estimated $\widehat{\beta}_{k\ell,k\ell^\prime}=h^{-1}\left(\bm{\theta}_{Z,k\ell}\right) - h^{-1}\left(\bm{\theta}_{Z,k\ell^\prime}\right)$, by the size of the corresponding blocks used in its estimation. Formally we compute the weight
\begin{eqnarray}
\widehat{\omega}_{k\ell,k\ell^\prime} = \frac{n_{0,k}n_{0,\ell}n_{1,\ell^\prime} + n_{1,k}n_{1,\ell}n_{0,\ell^\prime} }{n_k n_\ell n_{\ell^\prime}}
\end{eqnarray}
where $n_k =\sum_{i=1}^n \mathbf{1}_{\lbrace \widehat{\bm{\xi}}_i=k \rbrace} $.
We then consider the pairs in the set $\Omega = \lbrace (\ell,\ell^\prime), \ell,\ell^\prime \in \lbrace 1,\cdots , 2\widehat{K}\rbrace \vert \widehat{\psi}_\ell = \widehat{\psi}_{\ell^\prime} \rbrace $ and estimate the weighted sum
\begin{eqnarray}
\widehat{\beta}^{wa} = \frac{1}{\widehat{K}\vert \Omega \vert }\sum_{k=1}^{\widehat{K} }\sum_{(\ell,\ell^\prime )\in \Omega } \widehat{\omega}_{k\ell,k\ell^\prime} \widehat{\beta}_{k\ell,k\ell^\prime},
\label{eq:weightedmean_estimator}
\end{eqnarray}


Both estimators work well in practice, and we show some evidence in the examples and simulations.

\section{Simulation results}\label{section:simulations}
\subsection{Comparison with Variational EM}
We compare our spectral methods to a standard algorithm used in the literature, the variational EM algorithm, as implemented in the R package \texttt{blockmodels}. Our methods are implemented in the package \texttt{grdpg}, available on Github at \url{https://github.com/meleangelo/grdpg}. All the replication files for simulations and empirical application are at \url{https://github.com/meleangelo/grdpg_supplement}. We also note that the variational EM algorithm uses parallelization to increase computational efficiency, while our method is implemented without any parallelization and for networks with thousands of nodes.

\subsubsection{Example 1 (No covariates)}
In our first example, we do not include any covariates and we assume $h$ is the identity function, so that the link probabilities are defined by $\bm{P}_{ij}=\bm{X}_i^T \bm{X}_j$. We simulate networks with $n=2000, 5000, 10000,$ and $20000$  nodes, with latent space dimension $d=1$. In Table~\ref{tab:example1_K2} we report the results for $K=2$, with block centers $[p,q]=[0.1,0.7]$, and matrix of probabilities 
\begin{equation}
    \bm{\theta} = \left[\begin{matrix}
     0.01 & 0.07 \\
     0.07 & 0.49
    \end{matrix}\right].
\end{equation}

For simplicity, we assume that blocks are equally likely, that is $(\pi_1,\pi_2 )= (0.5,.0.5)$. 
To evaluate the performance of the algorithms, we compare clustering accuracy and computational time. The assignment of nodes to the correct block is summarized by the Adjusted Rand Index (ARI) \citep{RandARI1971}, and the computational time is given by the CPU time in seconds.   
Our point estimates are shown in Table~\ref{tab:example1_K2}, and below we report the estimated block probabilities for $n=2000$.
\begin{equation}
    \widehat{\bm{\theta}}_{VEM} = \left[\begin{matrix} 
    0.01002 & 0.07012\\
 0.07012 & 0.49086\\
    \end{matrix}
    \right],  \hspace{1cm}    \widehat{\bm{\theta}}_{GRDPG} = \left[\begin{matrix}
0.00998 & 0.07001 \\
0.07001 & 0.49091    \end{matrix}
    \right].
\end{equation}
The values $\hat{p}, \hat{q}$ shown in in Table~\ref{tab:example1_K2}  are obtained by singular value decomposition of the estimated probability matrix (and rotation).
\begin{table}
    \centering
    \caption{Point Estimates and CPU time for example 1 ($K=2$)}
    \begin{tabular}{l|rcccccrc}
         Estimator & $n$ & $K$  & $p$ & $\hat{p}$ & $q$ & $\hat{q}$ &  CPU Time (s) & ARI \\
              \hline \hline
        GRDPG & 2000 & 2 & 0.1 & 0.09993  & 0.7 & 0.70065 &   1.513  & 1 \\
        VEM   & 2000 & 2 &  0.1 & 0.10008 & 0.7 & 0.70061 & 39.679 & 1 \\
        GRDPG & 5000 & 2 & 0.1 & 0.10004  & 0.7 & 0.69977  & 8.548  & 1 \\
        VEM   & 5000 & 2 &  0.1 & 0.10008 & 0.7 & 0.69975  &  593.203 & 1 \\
        GRDPG & 10000 & 2 & 0.1 & 0.09994  & 0.7 & 0.69988  & 32.169 & 1 \\
        VEM   & 10000 & 2 &  0.1 & 0.09996 & 0.7 & 0.69987  &  4171.218 & 1 \\
        GRDPG & 20000 & 2 & 0.1 & 0.09998  & 0.7 & 0.70005  & 128.633 & 1 \\
       VEM   & 20000 & 2 &  NA &  &  &   &   \\
        GRDPG & 30000 & 2 & 0.1 & 0.09998  & 0.7 & 0.69995  & 386.210 & 1 \\
\hline\hline                       
    \end{tabular}
    \label{tab:example1_K2}
\end{table}
We notice that the VEM and GRDPG estimators produce similar point estimates and very precise clustering
of the nodes, as indicated by the ARI. However, our GRDPG estimator converges much faster than the VEM. For networks with $n=10000$ nodes, our method provides estimates in approximately 30 seconds, while the VEM takes more than one hour to converge to the final approximation. When $n=20000$ and $n=30000$, our GRDPG approach converges in about 2 minutes and less than 7 minutes, respectively, while the VEM is impractical.

Here, a crucial choice is the number of dimensions for the spectral embedding. In our simulation we know that the rank of the matrix $\bm{\theta}$ is 1, therefore this is the optimal dimension (see \cite{AthreyaEtAl2018a}). We choose $\hat{d}$ by profile likelihood methods as in \cite{ZhuGhodsi2006}.\footnote{The screeplot, not shown, displays a huge step down in the (absolute) value of the eigenvalues of the adjacency matrix at the largest eigenvalue, which suggests that 1 dimension is sufficient to approximate the structure of the adjacency matrix.}

The clustering of the latent positions in blocks is performed using the MCLUST method implemented in the package \texttt{Mclust} in R \citep{FraleyRaftery1999}.

In Table \ref{tab:example1_K5} we report results from the same model with $K=5$ and latent positions $\bm{\nu}=(0.1,0.3,0.5,0.7,0.9)$. The results are comparable to the previous table, our estimator scales very well with the size of the network, while obtaining the same point estimates of the VEM algorithm. In this example, the difference in scaling for the two estimators is more pronounced. In particular, going from $K=2$ to $K=5$ blocks does not increase the computational burden too much for the GRDPG-based estimator.
\begin{table}
    \centering
    \caption{Point Estimates and CPU time for example 1  ($K = 5$)}
    \begin{tabular}{l|rccccccrc}
                   &     &      &\multicolumn{5}{c}{latent positions/blocks} \\
         Estimator & $n$ & $K$  & 0.1  & 0.3 & 0.5 & 0.7 & 0.9 &   CPU Time (s) & ARI \\
              \hline \hline
        GRDPG & 2000 & 5 & 0.09976 & 0.30122 & 0.49819 & 0.70027 & 0.89987 &   1.543  & 1 \\
        VEM   & 2000 & 5 &  0.09994 & 0.30133 & 0.49825 & 0.70018 & 0.89973 & 257.713 & 1 \\
        GRDPG & 5000 & 5 & 0.10015 & 0.29952 & 0.49994  & 0.69962 & 0.90003  & 7.982 & 1 \\
        VEM   & 5000 & 5 & 0.10021  & 0.29958 & 0.49996 & 0.69958  &0.89999 & 926.330  & 1 \\
        GRDPG & 10000 & 5 & 0.09982 & 0.29975 & 0.49990 & 0.70006 & 0.90006 & 44.659 & 1 \\
        VEM   & 10000 & 5 & 0.09985 & 0.29977 & 0.49990 & 0.70004 & 0.90004 &  8128.253 & 1 \\
        GRDPG & 20000 & 5 & 0.10000 & 0.30001 & 0.50005 & 0.70019 & 0.89999 & 186.073 & 1 \\
        VEM   & 20000 & 5 & NA  &  &  &   &  \\
\hline\hline                        
    \end{tabular}
    \label{tab:example1_K5}
\end{table}

\subsubsection{Example 2 (logit link and binary covariate)}
We consider a model with a binary nodal covariate, $\bm{Z}_i \sim Bernoulli(0.5)$ and link probabilities
\begin{equation}
    \log\left(\frac{\bm{P}_{ij}}{1-\bm{P}_{ij}}\right) =\bm{X}_i^T \bm{X}_j + \beta \mathbf{1}_{\lbrace \bm{Z}_i=\bm{Z}_j \rbrace}.
\end{equation}

In this example we use $[p,q]=[-1.5, 1]$ and $\beta=1.5$, thus the matrix $\bm{\theta}$ is
\begin{equation}
    \text{logit}(\bm{\theta}) = \left[\begin{matrix}
     2.25 & -1.5 \\
     -1.5 & 1
    \end{matrix}\right]
\end{equation}
while the full matrix $\bm{\theta}_Z$ that includes the effect of covariates is 
\begin{equation}
    logit(\bm{\theta}_Z )= \left[\begin{matrix}
     3.75 & 2.25 & 0.00 & -1.50 \\
     2.25 & 3.75 &-1.50 & 0.00 \\
     0.00 &-1.50 & 2.50 & 1.00 \\
     -1.50&  0.00&  1.00&  2.50 
    \end{matrix}\right].
\end{equation}
We choose $\hat{d}$ by profile likelihood \citep{ZhuGhodsi2006}. In Figure~\ref{fig:est_latentpos_example2} we show the screeplots. In the upper-left, we display the screeplot for the adjacency matrix, which suggests using $\hat{d}=4$ as an estimate of the dimension for $\widehat{\bm{Y}}$. We note that the fourth largest eigenvalue is negative, and the GRDPG model takes this into account. In the center-left plot, we show the screeplot of the adjacency matrix after \emph{netting out the effect of the covariates}, which suggests the estimate $\hat{d}=1$ for determining the dimension of the unobserved latent positions $\bm{X}$.
\begin{figure}
    \centering
    \caption{Screeplots (upper left and center left), Estimated latent positions $\widehat{\bm{Y}}$ (right, only 2 dimensions out of 4 per plot) and estimated latent positions $\hat{\bm{X}}$, that is $\hat{p}$ and $\hat{q}$ in Example 2 (bottom left, up to orthogonal transformation) for $n=2000$.}
    \label{fig:est_latentpos_example2}
    \includegraphics[scale=0.6]{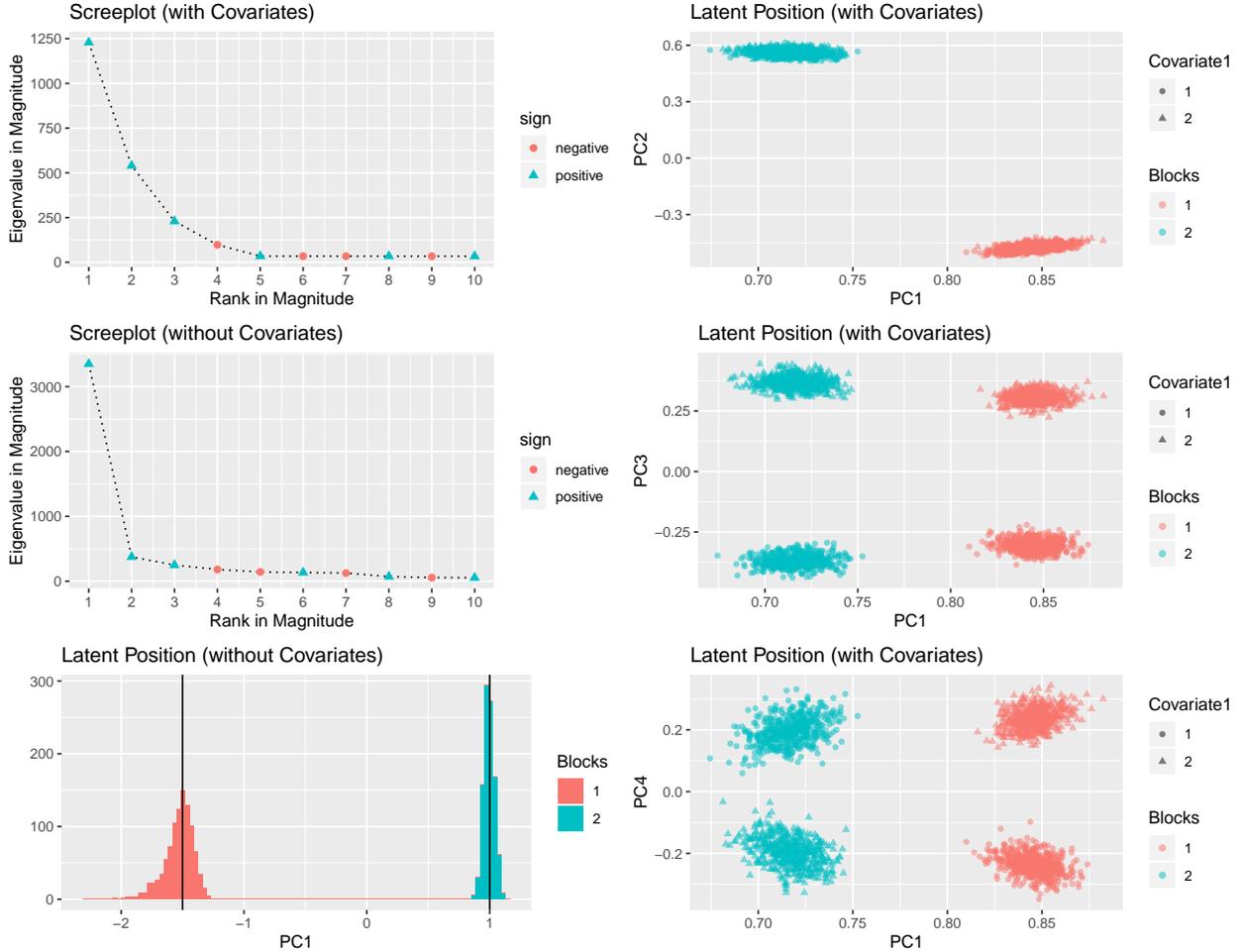}
\end{figure}

The point estimates for $\bm{\theta}_Z$ (up to a permutation of the block labels) when the network has $n=2000$ nodes are respectively
\begin{equation}
    \text{logit}(\widehat{\bm{\theta}}_{Z,VEM}) =  \left[\begin{matrix} 
       3.7443 & 2.2410 & -0.00367 & -1.5069\\
        2.2410 & 3.7443 &  -1.5069 & -0.0036 \\
    -0.00367 & -1.5069 & 2.5013 & 0.9980  \\
     -1.5069 & -0.0036 &  0.9980 & 2.5013
    
\end{matrix}
    \right],  
\end{equation}
\begin{equation}
\widehat{\bm{B}}_Z = \text{logit}( \widehat{\bm{\theta}}_{Z,GRDPG}) = \left[\begin{matrix} 
3.7762 & 2.2336 & -0.0062 & -1.5095 \\
 2.2336 &  3.7821 & -1.5007 &  -0.0042\\
-0.0062 & -1.5007 &  2.4979 & 0.9985\\
-1.5095 & -0.0042 & 0.9985 & 2.5045\\
    \end{matrix}
    \right].
    \label{eq:grdpg_estimate}
\end{equation}

According to our procedure, there are several ways to obtain an estimate of $\beta$. From matrix~(\ref{eq:grdpg_estimate}), we group rows 1 and 2 in one block, and rows 3 and 4 in another block by clustering the diagonal entries. We can get an estimate of $\beta$ as $\widehat{\bm{B}}_{Z,11}-\widehat{\bm{B}}_{Z,12}$ or $\widehat{\bm{B}}_{Z,22}-\widehat{\bm{B}}_{Z,21}$ or $\widehat{\bm{B}}_{Z,33}-\widehat{\bm{B}}_{Z,34}$, etc. We know by our theorem that each of these estimators is asymptotically normal. Instead of choosing which entries to use to estimate $\beta$, we pool all possible estimates, weighting them by the proportion of observations that are assigned to each block. For example, the estimate $\widehat{\bm{B}}_{Z,11}-\widehat{\bm{B}}_{Z,12}$ is weighted by the proportion of links in the network that are used to estimate it.  

The point estimates for $\beta$ reported in Table~\ref{tab:example2_K2} are $\widehat{\beta}_{GRDPG}=1.51201$ and $\widehat{\beta}_{VEM} = 1.50335$.
The estimated latent positions are $\hat{p}=-1.49712$ and $\hat{q}=1.00067$ for the VEM;
and $\hat{p}=-1.49454$ and $\hat{q}=0.99926$ for the GRDPG estimator. However, it takes almost 2 hours to obtain the VEM results, while it only takes 7 seconds with our estimator.
The left plots in Figure~\ref{fig:est_latentpos_example2} show the latent positions $\widehat{\bm{Y}}$ of the GRDPG (including the effect of covariates) estimated by ASE. We plot the first coordinate against each of the other three. In the second and third plot from the top, we can notice that the latent positions nicely cluster into 4 blocks, as our theory predicts. In the bottom-left plot in Figure~\ref{fig:est_latentpos_example2} we display the estimated latent positions $\widehat{\bm{X}}$, estimated by netting out the effect of the covariates. 
The figure shows how the estimated latent positions $\widehat{\bm{X}}$ cluster around the true values $p$ and $q$ (the black vertical lines).

As explained above, a central advantage of our approach is computational speed. Indeed, in Table~\ref{tab:example2_K2} we show that when we increase the size of the network to $n=5000$, the estimated parameters are essentially the same for VEM and GRPDG. However, the GRDPG estimator take less than 30 seconds to converge; the VEM estimate takes almost 10 hours.

\begin{table}
    \caption{Point Estimates and CPU time for example 2 ($K=2$)}
    \centering
    \begin{tabular}{l|rcccccccrc}
       Estimator & $n$ & $K$  & $p$ & $\hat{p}$ & $q$ & $\hat{q}$ & $\beta$ & $\hat{\beta}$ & CPU Time (s) & ARI \\
              \hline \hline
        GRDPG & 2000 & 2  & -1.5 & -1.49744 & 1 & 1.00077 & \multicolumn{2}{c}{no covariates} & 4.672 & 1\\
        VEM   & 2000 & 2 &  -1.5 & -1.49712 & 1 & 1.00067 & \multicolumn{2}{c}{no covariates} & 48.619  & 1\\
        GRDPG & 2000 & 2 &  -1.5 & -1.49454 & 1 & 0.99926 & 1.5 & 1.51201 & 7.557 & 1\\
        VEM & 2000 & 2 &  -1.5 & -1.49712  & 1 & 1.00067  & 1.5 & 1.50335  & 6903.673 & 1\\
        GRDPG & 5000 & 2  & -1.5 & -1.50029 & 1 & 1.00030 & \multicolumn{2}{c}{no covariates} & 17.539 & 1\\
        VEM   & 5000 & 2 &  -1.5 & -1.50019 & 1 & 1.00024 & \multicolumn{2}{c}{no covariates} & 537.831  & 1\\
        GRDPG & 5000 & 2 &  -1.5 & -1.49995 & 1 & 1.00064 & 1.5 & 1.49981  & 27.312  & 1 \\
        VEM & 5000 & 2 &  -1.5 & -1.50019  & 1 & 1.00024  & 1.5 & 1.49955  &  35331.012 & 1 \\
        GRDPG & 10000 & 2 &  -1.5 & -1.49989  & 1 & 1.00029 & \multicolumn{2}{c}{no covariates} & 55.428  & 1 \\
        GRDPG & 10000 & 2 &  -1.5 & 1.49992  & 1 & 0.99992 & 1.5 & 1.50190  & 91.067 & 1 \\

        \hline\hline
    \end{tabular}
    \label{tab:example2_K2}
\end{table}
\subsubsection{Example 3 (logit link, binary covariate and $d=2$)}
In Table~\ref{tab:example4_K2_d2} we show estimates for models with latent positions $\bm{\nu}_1=(-1.5,-1.0)$ and $\bm{\nu}_2=(1.0,0.5)$. For the simulations in the first 3 rows we set $\beta=1.5$. It is quite remarkable that the computational time does not increase much, with respect to the case of $d=1$. 

The second group of three rows shows the results of simulations with smaller $\beta=0.5$. This makes the estimation of the covariate effect more challenging. Indeed when $n=2000$ the classification in blocks and the point estimate are imprecise, as indicated by the adjusted Rand index (ARI). When we increase the network size to $n=5000$ and $n=10000$, the accuracy of the point estimates improves significantly. This example shows that our approach is extremely useful in very large networks, where VEM may become computationally impractical.
      
\begin{table}
    \caption{Point estimates and time for Example 3.}
    \centering
    \begin{tabular}{l|rcccrc}
       Estimator & $n$ & $K$  & $\beta$ & $\hat{\beta}$  & Time (s) & ARI \\
              \hline \hline
GRDPG & 2000 & 2 &  1.5 & 1.51760  & 7.335 & 1\\
GRDPG & 5000 & 2 &  1.5 & 1.49946  & 28.153 & 1  \\
GRDPG & 10000 & 2 & 1.5 & 1.50257  &  99.128 & 1    \\
\hline
GRDPG & 2000 & 2 &  0.5 & 0.64145 & 7.815 & 0.998\\
GRDPG & 5000 & 2 &  0.5 &  0.56222 & 26.763 & 1  \\
GRDPG & 10000 & 2 & 0.5 & 0.51617  &  88.562 & 1     \\

        \hline\hline
    \end{tabular}
    \label{tab:example4_K2_d2}
\end{table}


In summary, our simple examples and simulations show that our GRDPG-based estimator is quite fast and scales well to large networks. These good computational properties are obtained without sacrificing the accuracy of the estimates, as we prove that the algorithm produces the same point estimates as the variational EM in all the examples. 

\subsection{Monte Carlo experiments}
Given the computational times shown in the previous section, we run a simple Monte Carlo to understand the bias of the spectral estimator in networks of moderate size.\footnote{We do not compare the spectral estimator to the variational EM estimator, because the latter is too slow for a Monte Carlo with 1000 repetitions, even after parallelizing the execution. }
We estimate a model with two binary observed covariates, 
\begin{equation}
    \bm{Z}_i \sim Bernoulli(b_z) \ \ \ and \ \ \  \bm{W}_i \sim Bernoulli(b_w) 
\end{equation}
and vary the probabilities $b_z$ and $b_w$, as well as the correlation among the two variables. We estimate the following model in each Monte Carlo design
\begin{equation}
    \log\left(\frac{\bm{P}_{ij}}{1-\bm{P}_{ij}}\right) =\bm{X}_i^T \bm{X}_j + \beta_1 \mathbf{1}_{\lbrace \bm{Z}_i=\bm{Z}_j \rbrace} + \beta_2 \mathbf{1}_{\lbrace \bm{W}_i=\bm{W}_j \rbrace}.
    \label{eq:mc_model}
\end{equation}
The Monte Carlo design considers networks of sizes $n=2000, 5000, 10000$ and we set the number of blocks to $K=2$. For all the simulations the parameter value that generates the data is $\bm{\beta} = (0.5,0.75)$ and the centers of the blocks are $\bm{\nu}=(-1.5,1.0)$. We summarize the designs in Table \ref{tab:montecarlodesign}. For each design and network size, we simulate 1000 networks and estimate the parameters of model \eqref{eq:mc_model} with the simple mean estimator and the weighted mean estimator.
\begin{table}[!h]
    \centering
    \caption{Monte Carlo design}
    \begin{tabular}{l|ccccc}
    \hline\hline
    Design & $\pi_1$ & $b_z$ & $b_w$ & correlation \\
    \hline\hline 
      1   & $0.5$  &  $0.5$ & $0.5$ &  independent \\
      2   & $0.5$  &  $0.5$ & $0.5$ &  $0.3$ \\
      3   & $0.3$  &  $0.5$ & $0.5$ &  independent \\
      4   & $0.3$  &  $0.4$ & $0.6$ &  independent \\
      5   & $0.3$  &  $0.4$ & $0.6$ &  $0.3$ \\
      \hline\hline
    \end{tabular}
    
    \label{tab:montecarlodesign}
    \flushleft \scriptsize In the Monte Carlo simulations we set the number of blocks $K=2$ with centers $\bm{\nu}=(-1.5,1)$; the true parameter vector for the covariates is $\bm{\beta}=(0.5,0.75)$ and we repeat the simulations for $n=2000,5000,10000$. We vary the probability of belonging to the first block $\pi_1$, the probability of the Bernoulli covariates $b_z$ and $b_w$, respectively for $\bm{Z}_i$ and $\bm{W}_i$; and we allow the covariates to have correlations in some of the simulations. Each Monte Carlo is the outcome of 1000 simulations.
\end{table}

The first design corresponds to the examples in the previous section. Blocks are assumed to have same size and observables are independent Bernoulli variables with equal probability. The second design introduces correlation among the observables, as this is a realistic feature of many datasets. Designs 3 and 4 are intended to test the effect of unbalanced block size and unbalanced covariates, respectively, while maintaining the assumption of independence among observables. The final design assumes that we have unbalanced blocks, unbalanced covariates and correlated observables, allowing us to understand how the estimator behaves in a realistic setting. We expect that unbalancedness and correlation will increase bias, but this problem is less severe for larger networks, as our theory shows that the bias becomes vanishingly small as we increase the size of the network.


\begin{table}[]
    \centering
   \caption{Monte Carlo results, simple mean estimator \eqref{eq:simplemean_estimator}}
\begin{tabular}{lccccc}
  \hline
  $n$  & $\vert \hat{\beta}_1 - \beta_1\vert$ & mcse & $\vert \hat{\beta}_2 - \beta_2\vert$ & mcse & time \\ 
  \hline
&  \multicolumn{5}{c}{\underline{\textbf{Design 1}}}\medskip\\
 2000 & 0.0576 & 0.0001696 & 0.0350 & 0.0001366 & 12.35 \\ 
 5000 & 0.0134 & 0.0000311 & 0.0082 & 0.0000281 & 31.68 \\ 
 10000 & 0.0016 & 0.0000093 & 0.0009 & 0.0000092 & 137.97 \\    
 \hline
&  \multicolumn{5}{c}{\underline{\textbf{Design 2}}}\medskip\\
 2000 & 0.1118 & 0.0008083 & 0.0843 & 0.0007631 & 9.14 \\ 
 5000 & 0.0240 & 0.0000475 & 0.0168 & 0.0000413 & 32.75 \\ 
 10000 & 0.0040 & 0.0000109 & 0.0026 & 0.0000104 & 140.20 \\  
 \hline
&  \multicolumn{5}{c}{\underline{\textbf{Design 3}}}\medskip\\
 2000 & 0.1683 & 0.0008531 & 0.1039 & 0.0006886 & 12.71 \\ 
 5000 & 0.0201 & 0.0000424 & 0.0089 & 0.0000379 & 52.91 \\ 
 10000 & 0.0015 & 0.0000140 & 0.0011 & 0.0000124 & 142.54 \\     
 \hline
&  \multicolumn{5}{c}{\underline{\textbf{Design 4}}}\medskip\\
 2000 & 0.2144 & 0.0016978 & 0.1697 & 0.0014046 & 23.83 \\ 
 5000 & 0.0399 & 0.0001094 & 0.0228 & 0.0000921 & 51.71 \\ 
 10000 & 0.0046 & 0.0000162 & 0.0016 & 0.0000143 & 138.92 \\
   \hline
&  \multicolumn{5}{c}{\underline{\textbf{Design 5}}}\medskip\\
 2000 & 0.1607 & 0.0012640 & 0.1714 & 0.0010800 & 14.98 \\ 
 5000 & 0.0738 & 0.0007459 & 0.0949 & 0.0009128 & 54.29 \\ 
 10000 & 0.0498 & 0.0002895 & 0.0459 & 0.0003064 & 138.06 \\ \hline\hline
\end{tabular}
     \label{tab:MonteCarlo_results_simplemean}
         \flushleft \scriptsize Monte Carlo simulations for the weighted mean estimator. Each Monte Carlo consists of 1000 simulated networks, with $K=2$ unobserved blocks and $\bm{\beta}=(0.5,0.75)$. The centers of the blocks are $\bm{\nu}=(-1.5,1.0)$. Designs are summarized in Table \ref{tab:montecarlodesign}.
\end{table}

The results of our simulations are reported in Tables \ref{tab:MonteCarlo_results_simplemean} (simple mean estimator) and \ref{tab:MonteCarlo_results_weigthedmean} (weighted mean estimator). 
For each design, we report the absolute difference between the estimated parameter and the true value, the Monte Carlo standard error and the average time for estimation.\footnote{The time of estimation reported in the table includes the following steps: 1) compute the ASE from the adjacency matrix; 2) compute the matrix of latent positions; 3) cluster latent positions to recover blocks; 4) compute matrix $\widehat{\bm{B}}_Z$; 5) cluster diagonal entries of matrix $\widehat{\bm{B}}_Z$ to recover the unobservable block structure; 6) estimate $\bm{\beta}$ using the information on the block structure and the entries of matrix $\widehat{\bm{B}}_Z$; 7) compute simple mean and weighted mean of the estimated $\widehat{\bm{\beta}} $ according to \eqref{eq:simplemean_estimator} and \eqref{eq:weightedmean_estimator}. The simulation takes a little longer because we need to generate the data and the adjacency matrices for the Monte Carlo. Code is available in Github.} When using the simple mean estimator, the estimates are precise, while displaying a small bias. The most challenging design for our estimator is Design 5, where we impose different unobserved block size, different Bernoulli probabilities for the observables and correlation among observed characteristics. As expected, these features increases the bias in our estimates; however, this problem becomes less severe with larger networks.  

Our weighted mean estimator has similar behavior. 

\begin{table}[]
    \centering
   \caption{Monte Carlo results, weighted mean estimator \eqref{eq:weightedmean_estimator}}
\begin{tabular}{lccccc}
  \hline
  $n$  & $\vert \hat{\beta}_1^w - \beta_1\vert$ & mcse & $\vert \hat{\beta}_2^w - \beta_2\vert$ & mcse & time \\ 
  \hline
&  \multicolumn{5}{c}{\underline{\textbf{Design 1}}}\medskip\\
 2000 & 0.0650 & 0.0022329 & 0.0466 & 0.0031385 & 12.35 \\ 
 5000 & 0.0134 & 0.0000311 & 0.0082 & 0.0000281 & 31.68 \\ 
 10000 & 0.0016 & 0.0000093 & 0.0009 & 0.0000092 & 137.97 \\ 
   \hline
&  \multicolumn{5}{c}{\underline{\textbf{Design 2}}$^\ast$}\medskip\\
 2000 & 0.2272 & 0.0084645 & 0.1942 & 0.0082881 & 9.14 \\ 
 5000 & 0.0239 & 0.0000473 & 0.0168 & 0.0000411 & 32.75 \\ 
 10000 & 0.0040 & 0.0000109 & 0.0026 & 0.0000104 & 140.20 \\  
  \hline
&  \multicolumn{5}{c}{\underline{\textbf{Design 3}}}\medskip\\
 2000 & 0.0612 & 0.0002496 & 0.0409 & 0.0002106 & 12.71 \\ 
 5000 & 0.0095 & 0.0000230 & 0.0054 & 0.0000224 & 52.91 \\ 
 10000 & 0.0013 & 0.0000082 & 0.0007 & 0.0000075 & 142.54 \\  
 \hline
&  \multicolumn{5}{c}{\underline{\textbf{Design 4}}}\medskip\\
 2000 & 0.2757 & 0.0068607 & 0.2344 & 0.0061783 & 23.83 \\ 
 5000 & 0.0155 & 0.0002438 & 0.0098 & 0.0002319 & 51.71 \\ 
 10000 & 0.0025 & 0.0000085 & 0.0014 & 0.0000080 & 138.92 \\  
   \hline
&  \multicolumn{5}{c}{\underline{\textbf{Design 5}$^\ast$}}\medskip\\
 2000 & 1.0342 & 0.0058407 & 0.9197 & 0.0052219 & 14.98 \\ 
 5000 & 0.9111 & 0.0040518 & 0.8047 & 0.0037239 & 54.29 \\ 
 10000 & 0.1574 & 0.0032735 & 0.1428 & 0.0030875 & 138.06 \\ 
\hline\hline
\end{tabular}
     \label{tab:MonteCarlo_results_weigthedmean}
     \flushleft \scriptsize Monte Carlo simulations for the weighted mean estimator. Each Monte Carlo consists of 1000 simulated networks, with $K=2$ unobserved blocks and $\bm{\beta}=(0.5,0.75)$. The centers of the blocks are $\bm{\nu}=(-1.5,1.0)$. Designs are summarized in Table \ref{tab:montecarlodesign}.\\
     $^\ast$ For Design 5, in few cases the estimation did not converge, so we eliminated those simulations from the computation of the mean.
\end{table}

\subsection{Practical estimation of standard errors}
In the empirical application we estimate the standard errors for the observed covariates effects according to the formula in Theorem \ref{thm:clt_general_h_block_estim}, using a plug-in estimate. To understand the behavior of this plug-in estimator, we ran several Monte Carlo experiments. In Table \ref{tab:mc_variance} we report results using the model with independent covariates and balanced blocks (Design 1 in Table \ref{tab:montecarlodesign}).

In columns 1 and 2 we show the true standard error computed using the formula in Theorem \ref{thm:clt_general_h_block_estim} for the simple mean estimator; in columns 3 and 4, the standard error is computed using the weighted estimator with the true proportions of each block; in columns 5 and 6 we have our weighted estimator, using the estimated parameters and block proportions. Our plug-in estimator is very conservative. For a network with $5000$ nodes, the estimated standard error is extremely high, when compared to the true value. Even when the network size is increased to $10000$ nodes, the plug-in estimate for the standard error is quite large.

We conclude that using a plug-in estimator usually overestimates the standard errors. Our estimates are thus very conservative.

\begin{table}[!t]
    \centering
     \caption{Monte Carlo results for plug-in estimator of the standard error}
    \begin{tabular}{lcccccc}
    \hline\hline
    $n$ & $se(\beta_1)$ & $se(\beta_2)$ & $se(\beta_1^w)$ & $se(\beta_2^w)$ & $se(\hat{\beta}_1^w)$ & $se(\hat{\beta}_2^w)$ \\
    \hline
      5000 & 0.0574 & 0.0532 & 0.0208 & 0.0194 & 1.5819 & 1.5183  \\
      10000 & 0.0286 & 0.0265 & 0.0104 & 0.0097 & 0.8443 & 0.8072\\ 
      \hline \hline
    \end{tabular}
    \label{tab:mc_variance}
    \flushleft \scriptsize In the Monte Carlo simulations we set the number of blocks to $K=2$. We run 100 simulations, with $n=5000, 10000$. The true parameters are $\beta = (0.5,0.75)$ and the observable covariates are binary and independent. Both blocks and covariates are balanced, so each block has probability $0.5$ and the covariates are independent  Bernoulli variables with success probability $0.5$. In columns 1 and 2 we show the standard error computed using the formula in Theorem \ref{thm:clt_general_h_block_estim}; in columns 3 and 4, the standard error is computed using the weighted estimator with the true proportions of each block; in columns 5 and 6 we have our weighted estimator. 
\end{table}

\section{Application to Facebook 100 dataset}
We apply our method to study a network of Facebook friendships, using the Facebook~100 dataset from \cite{TraudEtAl2012}.\footnote{The entire dataset is available at \texttt{https://archive.org/details/oxford-2005-facebook-matrix}.} This network was extracted from the Facebook platform in September 2005, providing a snapshot of the friendship among students, faculty, staff and alumni at 100 U.S. universities.

We perform an analysis similar to \cite{AtchadeEtAl2019}, using the Harvard University network data. The dataset consists of 15126 nodes and 7 nodal covariates: role, gender, major, minor, dorm, year, and high school. These are all discrete variables. We focus on dorm, gender, and role in the analysis. We make each of these variables binary. So rather than specify specific dorm information, our control variable indicates whether the student lives on or off-campus. The role is binarized to indicate whether the node is a student or not.\footnote{Roles include students, faculty, staff, alumni, etc. We focus on students because they are the ones mostly using the platform in 2005.} 

\begin{table}[h]
    \centering
       \caption{Descriptive statistics of Harvard University's Facebook network.}
    \begin{tabular}{l|cc}
    \hline\hline
         & Original    & Sample for\\
         &          & estimation\\
         \hline \hline
      $n$   & 15126 &  13003\\
      avg. degree   & 109.03 & 105.02\\
      clust. coeff.  & 0.135 &  0.137\\
      prop. female &  0.466 &  0.539\\
      prop. students & 0.508 &  0.520\\
      prop. off-campus & 0.229 &  0.153\\
      \hline \hline
    \end{tabular}
    \label{tab:Harvard_network_stats}
    \flushleft \scriptsize The first column is the original network data. The second column is the largest connected component of the network including all nodes with non-missing gender information and graduating in any year.
\end{table}

The characteristics of the network are shown in Table~\ref{tab:Harvard_network_stats}. In the first column we report the descriptive statistics for the original data, containing $n=15126$ nodes, with an average degree of $109.03$, 
an average clustering coefficient of $0.135$,
with $46.6\%$ females, $50.8\%$ students and $22.9\%$ of people living off-campus. The average degree and clustering coefficient suggest that this is a moderately sparse network.  

The second column contains the descriptive statistics for our processed sample. 
We keep all nodes with non-missing gender information; once we delete all missing gender information, there are no missing values for the other covariates;  
we then compute the largest connected component of the resulting network.\footnote{This is a standard procedure in the literature on SBMs \citep{AtchadeEtAl2019, AthreyaEtAl2018a, Abbe2018}.} The final network contains a higher proportion of females ($53.9\%$) and students ($52\%$); and a smaller fraction of people living off-campus ($15.3\%$) than the original data. The average degree is slightly smaller ($105.02$), because we have deleted some nodes and edges; however, the clustering coefficients are of similar magnitude, $0.135$ and $0.137$, respectively.\footnote{Before we proceed to estimation, we regularize the adjacency matrix using the standard method proposed in \citep{LevinaRegularization2017}. This regularization step avoids numerical issues with the spectral decomposition arising from significant node degree heterogeneity.} 

We estimate the following model with three covariates
\begin{eqnarray}
    \log\left(\frac{\bm{P}_{ij}}{1-\bm{P}_{ij}}\right) &=& \bm{X}_i^T \bm{X}_j + \beta_1 \mathbf{1}_{\lbrace female_i=female_j \rbrace} \notag\\
    &+& \beta_2 \mathbf{1}_{\lbrace student_i=student_j \rbrace}+ \beta_3 \mathbf{1}_{\lbrace off-campus_i=off-campus_j \rbrace}.
    \label{eq:fb_model}
\end{eqnarray}
where $female_i=1$ if the node is female, $student_i=1$ if the node is a student and $off-campus_i=1$ if the node lives off-campus. 
We first estimate models with one binary covariate, using each of our control variables individually. Next we estimate the full model \eqref{eq:fb_model} with three controls. For the Adjacency Spectral Embedding we choose the dimension of the latent space $\hat{d}=2$, using the profile likelihood method in \cite{ZhuGhodsi2006}, as in our simulations.\footnote{Multiple methods exist for selecting the embedding dimension in practice, and this remains a topic of current research. In the context of networks, choosing a dimension smaller than the true $d$ will introduce bias in the estimated latent positions; on the other hand, using an embedding dimension larger than the true $d$ will increase the variance of the estimated latent positions. In this trade-off we prefer to err on the side of overestimating $d$. Specifically, we choose the value one plus the location of the first elbow in the screeplot.} 

The clustering of the estimated latent positions is performed with a Gaussian mixture model, using the \texttt{MCLUST} implementation of \cite{FraleyRaftery1999} in \texttt{R}. We obtain latent positions estimates and $\widehat{\tilde{K}}=32$ blocks from the adjacency matrix. We then obtain the estimated matrix $\hat{\bm{B}}_Z$, and cluster its diagonal entries to recover the (unobserved) blocks $\hat{K}$ and estimate the vector of parameters $\bm{\beta}$. 
\begin{table}[ht]
\centering
  \caption{Estimation results for Harvard University network data (Facebook~100)}
\begin{tabular}{l|cccc}
  \hline\hline
\textbf{Variable} & \textbf{1} & \textbf{2} & \textbf{3} & \textbf{4} \\ 
  \hline
female & 0.6614 &  &  & 0.5766 \\ 
   & (0.0240) &  &  & (0.0061) \\ 
  student &  & 0.6883 &  & 0.6463 \\ 
   &  & (0.0360) &  & (0.0144) \\ 
  off-campus &  &  & 0.3269 & 0.3315 \\ 
   &  &  & (0.0113) & (0.0020) \\ 
   \hline \hline
  $n$ & 13003 & 13003 & 13003 & 13003 \\ 
  $\hat{K}$ & 16 & 16 & 16 & 4 \\ 
  $\hat{d}$ & 2 & 2 & 2 & 2 \\ 
   \hline\hline
\end{tabular}
\label{tab:estimation_Harvard}
    \flushleft \scriptsize Parameter estimates for the effect of observable covariates using Harvard University network data from the Facebook-100 dataset. The point estimate is obtained with the weighted mean estimator. The number in parenthesis is the naive standard error estimate, using a plug-in estimator and the formula for variance from Theorem  \ref{thm:clt_general_h_block_estim}. All estimates are obtained using the first elbow of the screeplot \citep{ZhuGhodsi2006}.
\end{table}

The estimated parameters are shown in Table~\ref{tab:estimation_Harvard}. In the first three columns we report estimates for models with a single binary covariate. Each coefficient is precisely estimated, according to our naive plug-in standard error estimator; the estimated effects are all positive, which we interpret as evidence of homophily, a usual feature of many social networks. The point estimates are very similar when we estimate the full model \eqref{eq:fb_model} (Column 4). These results are also consistent with the analysis in \cite{TraudEtAl2012} and \cite{AtchadeEtAl2019}. 
Our method estimates $\widehat{K}=4$ unobserved blocks. If we choose to include only one covariate, the number of blocks estimated is $\widehat{K}=16$ (Columns 1--3). 

The results indicate that there are additional unobserved characteristics that affect the network formation among Facebook users. In particular, the four blocks may capture shared interests, common preferences, similar schedules and additional information that is unobservable to the researcher.  


In summary, our method is computationally faster than the traditional VEM, performs well in Monte Carlo simulations with different designs and data generating processes, and provides useful insights in real-data, empirical applications.

\section{Conclusion}
We have developed a spectral estimator for stochastic blockmodels with nodal covariates in large networks. The main theoretical contribution is an asymptotic normality result for the spectral estimator of the covariates' effect  on the probability of linking. Our work leverages the relationship between generalized random dot product graphs and stochastic blockmodels, extending existing frameworks to include observed covariates and constructing an estimator that is fast and scalable for large networks. Our theoretical results also apply to moderately sparse graphs, which is important in a host of applications in economics and more generally in social sciences, public health, and computer science, where network data are often viewed as being sparse.

We have provided examples demonstrating that our method delivers the same accuracy as the variational EM algorithm, while converging much faster. Our Monte Carlo simulations and the empirical application show that this method works best in very large networks, when the variational EM becomes impractical.  

We consider the present work a first step in the study of this class of models and the foundation for inference for SBMs and other latent position models for large networks with nodal covariates. While we have focused on binary and discrete covariates in this work, extensions to continuous covariates are currently being pursued via recently developed Latent Structure Models \citep{AthreyaLSM2018}. In future work, similar ideas can also be applied to directed networks and bipartite networks \citep{AKM1999, BonhommeLamadonManresa2018}, significantly expanding the realm of GRDPG applications in economics and social sciences.

\bibliographystyle{jmr}
\bibliography{biblioRDPG}

\appendix

\numberwithin{equation}{section}
\numberwithin{theorem}{section}
\numberwithin{lemma}{section}
\section{Proofs}
We first provide the general proof strategy for the simple case in which the block assignment function $\bm{\tau}$ is known. The next two theorems provide a foundation and roadmap for the proof for the more general case. 

\subsection{Blocks known}
The simplest case is when the latent block assignments are known, the value of $d$ and $K$ are known, and $h$ is the identity function.
Let $\bm{\nu}=(\bm{\nu}_1, \dots, \bm{\nu}_{K}) $ be the centers of the $K$ blocks $\bm{X}$, such that $i$ and $j$ belong to \emph{unobserved} block $k$ when $\bm{X}_i=\bm{X}_j=\bm{\nu}_k$. Let $\bm{\tau}$ be the function to assigns nodes to unobserved blocks, that is $\tau_i=k$ if $\bm{X}_i=\bm{\nu}_k$.
\emph{For this subsection we will assume that $\bm{\tau}$ is known}. Our model is
\begin{equation}
\bm{A}_{ij}\vert \bm{X}_i, \bm{X}_j, \bm{Z}_i, \bm{Z}_j,\beta \overset{ind}{\sim} Bernoulli \left(\bm{X}_i^T  \bm{X}_j + \beta \mathbf{1}_{\lbrace \bm{Z}_i = \bm{Z}_j \rbrace} \right)
\label{eq:model_known_tau}
\end{equation}
with $\bm{Z}_i \overset{ind}{\sim} Bernoulli(b_{\tau_i})$. Therefore, model~(\ref{eq:model_known_tau}) is a $\widetilde{K}=2K$ SBM, because we have $K$ unobserved blocks, each split in 2 by the observed covariates.
The probabilities of belonging to a block $k$ for this $\widetilde{K}$-block SBM are  $\bm{\eta} = (\eta_1, \dots, \eta_{\widetilde{K}}) = ( \pi_1 \cdot b_1, \pi_1 \cdot (1-b_1), \pi_2 \cdot b_2, \pi_2 \cdot (1-b_2), \dots, \pi_K \cdot b_K, \pi_K \cdot (1-b_K) )  $. Additionally, the assignment functions are $\bm{\xi} = (\xi_1, \dots,\xi_{n} )$, such that  $\xi_i = 1$ if $ \tau_i = 1$ and $ \bm{Z}_i=0 $; $\xi_i =2$ if $\tau_i=1$ and $ \bm{Z}_i=1$; $\xi_i=3$ if $\tau_i=2$ and $ \bm{Z}_i=0 $; $\xi_i=4$ if $\tau_i=2$ and $ \bm{Z}_i=1$; and so on.\\

Our SBM is $ (\bm{A}, \bm{\xi}, \bm{Z}) \sim SBM(\bm{\theta}_Z,\bm{\eta})$ with $\widetilde{K}\times \widetilde{K}$ matrix of probabilities $\bm{\theta}_Z$ 
\begin{tiny}
\begin{equation}
\bm{\theta}_Z = \bordermatrix{ & \tau=1;Z=0 & \tau=1;Z=1 & \tau=2;Z=0 & \tau=2;Z=1 & \cdots & \tau=K;Z=0 & \tau=K;Z=1 \cr
                \tau=1;Z=0 & \bm{\nu}_1^T  \bm{\nu}_1 + \beta & \bm{\nu}_1^T  \bm{\nu}_1  & \bm{\nu}_1^T  \bm{\nu}_2  + \beta & \bm{\nu}_1^T \bm{\nu}_2 & \cdots & \bm{\nu}_1^T  \bm{\nu}_K  + \beta & \bm{\nu}_1^T  \bm{\nu}_K \cr                
                \tau=1;Z=1 & \bm{\nu}_1^T  \bm{\nu}_1 & \bm{\nu}_1^T  \bm{\nu}_1 + \beta & \bm{\nu}_1^T \bm{\nu}_2 & \bm{\nu}_1^T  \bm{\nu}_2 + \beta & \cdots & \bm{\nu}_1^T  \bm{\nu}_K  & \bm{\nu}_1^T  \bm{\nu}_K  + \beta \cr
                \tau=2;Z=0 & \bm{\nu}_2^T  \bm{\nu}_1 + \beta & \bm{\nu}_2^T  \bm{\nu}_1  & \bm{\nu}_2^T \bm{\nu}_2  + \beta & \bm{\nu}_2^T  \bm{\nu}_2 & \cdots & \bm{\nu}_2^T  \bm{\nu}_K  + \beta & \bm{\nu}_2^T \bm{\nu}_K \cr                
                \tau=2;Z=1 & \bm{\nu}_2^T  \bm{\nu}_1 & \bm{\nu}_2^T \bm{\nu}_1 + \beta & \bm{\nu}_2^T  \bm{\nu}_2 & \bm{\nu}_2^T \bm{\nu}_2 + \beta & \cdots & \bm{\nu}_2^T \bm{\nu}_K  & \bm{\nu}_2^T \bm{\nu}_K  + \beta \cr
                \vdots &  \vdots & \vdots & & & \ddots \cr
                \tau=K;Z=0 & \bm{\nu}_K^T  \bm{\nu}_1 + \beta & \bm{\nu}_K^T  \bm{\nu}_1  & \bm{\nu}_K^T \bm{\nu}_2  + \beta & \bm{\nu}_K^T \bm{\nu}_2 & \cdots & \bm{\nu}_K^T \bm{\nu}_K  + \beta & \bm{\nu}_K^T  \bm{\nu}_K \cr                
                \tau=K;Z=1 & \bm{\nu}_K^T  \bm{\nu}_1 & \bm{\nu}_K^T  \bm{\nu}_1 + \beta & \bm{\nu}_K^T  \bm{\nu}_2 & \bm{\nu}_K^T  \bm{\nu}_2 + \beta & \cdots & \bm{\nu}_K^T  \bm{\nu}_K  & \bm{\nu}_K^T   \bm{\nu}_K  + \beta \cr
 }.
\label{eq:matrix_thetaZ_shape_appendix}
\end{equation}
\end{tiny}

The following theorem establishes asymptotic normality for the estimator $\hat{\beta}=\hat{\bm{\theta}}_{Z,11}-\hat{\bm{\theta}}_{Z,12}$.

\begin{theorem} Let $\bm{A}$ be an adjacency matrix from model (\ref{eq:model_known_tau}) with $h$ equal to the identity function $h(u)=u$. Let $\bm{\tau}$ be known. Then $\widehat{\beta} = \widehat{\bm{\theta}}_{Z,11}-\widehat{\bm{\theta}}_{Z,12} $ 
is asymptotically normal, that is
\begin{eqnarray}
n \left(\hat{\beta} - \beta - \frac{\psi_\beta}{n} \right)\overset{d}{\longrightarrow} N(0,\sigma_\beta^2)
\end{eqnarray}
where both $\psi_\beta$ and $\sigma_\beta^2$ are derived in the appendix.
\label{thm:clt_linear_block_known}
\end{theorem}
\begin{proof}
The results in the theorem exploit the fact that $\beta$ is a linear function of the entries of $\bm{\theta}_Z$, whose spectral estimators also exhibit asymptotic normality \citep{TangEtAl2017}. There is a small bias term that goes to zero with $n$, and we can use the result for inference. 

To prove the central limit theorem using the machinery of spectral estimation of generalized random dot product graphs models, we proceed in several steps. \\

\noindent \textbf{STEP 1: reformulate SBM as GRDPG}\\
Our SBM can be thought of as a GRDPG. Indeed, consider the eigendecomposition of matrix $\bm{\theta}_Z = \bm{U}\bm{\Sigma} \bm{U}^T$,  and define $\bm{\mu}=[\bm{\mu}_1,\bm{\mu}_2, \dots ,\bm{\mu}_{\widetilde{K}}]$ as the rows of $\bm{U}\vert \bm{\Sigma}\vert^{1/2}$, then $F=\sum_{k=1}^{\widetilde{K}}\eta_k \delta_{\bm{\mu}_k}$, where $\delta$ is the Dirac-delta; and $d_1$ and $d_2$ are the number of positive and negative eigenvalues of $\bm{\theta}_Z$, respectively. So the GRDPG corresponding to our SBM $ (\bm{A}, \bm{\xi}, \bm{Z}) \sim SBM(\bm{\theta}_Z,\bm{\eta})$ is given by
\begin{equation}
\bm{A}_{ij}\vert \bm{Y}_i,\bm{Y}_j \overset{ind}{\sim} Bernoulli(\bm{Y}_i^T \bm{I}_{d_1,d_2} \bm{Y}_j)
\end{equation}
where  $d_1+d_2=\widetilde{d}=rank(\bm{\theta}_Z)$ and $\bm{Y}$ is the $n\times \widetilde{d}$  vector of latent positions with centers $\bm{\mu}$.\\

\noindent\textbf{STEP 2: spectral estimation for GRDPG}\\
Letting $\widetilde{d}=rank(\bm{\theta}_Z)$, we then perform Adjacency Spectral Embedding (ASE) for $\bm{A}$ and obtain $\bm{A}=\bm{\widehat{U}}\bm{\widehat{S}}\bm{\widehat{U}}^T + \bm{\widehat{U}}_{\bot}\bm{\widehat{S}}_{\bot}\bm{\widehat{U}}_{\bot}^T$, where $\bm{\widehat{S}}$ is the diagonal matrix containing the $\widetilde{d}$ largest eigenvalues of $\bm{A}$ in absolute value, and $\bm{\widehat{U}}$ is the $n\times \widetilde{d}$ matrix whose columns are the corresponding eigenvectors of $\bm{A}$. Using a clustering procedure we can cluster the  estimated
latent positions $\bm{\widehat{Y}}=\bm{\widehat{U}}\vert \bm{\widehat{S}}\vert ^{1/2}$ (using K-means or GMM), obtaining $\widetilde{K}$ clusters and estimates of the clusters centers $\hat{\bm{\mu}}$ and cluster assignments $\bm{\hat{\xi}}$. Notice that these estimates are consistent \citep{AthreyaEtAl2018a,TangPriebe2018}.

Remember that in the present setting we assume $\bm{\tau}$ is known, so our estimates for the probabilities are
\begin{equation}
\bm{\widehat{\theta}}_{Z,k\ell}=  \hat{\bm{\mu}}_{k}^{T} \bm{I}_{d_1,d_2} \hat{\bm{\mu}}_\ell
\end{equation}
for any pair $k,\ell=1, \dots ,\widetilde{K}$. \\

\noindent\textbf{STEP 3: estimate $\beta$ from matrix $\bm{\widehat{\theta}}_Z$}\\
From the matrix $\bm{\theta}_Z$ we notice that $\beta = \bm{\theta}_{Z,11}- \bm{\theta}_{Z,12}$, thus we can use
the estimator $\hat{\beta} = \bm{\widehat{\theta}}_{Z,11}- \bm{\widehat{\theta}}_{Z,12}$. With some algebra we obtain
\begin{eqnarray}
\hat{\beta} &=& \bm{\widehat{\theta}}_{Z,11}- \bm{\widehat{\theta}}_{Z,12} \\
&=& \bm{\widehat{\theta}}_{Z,11} -\bm{\theta}_{Z,11} + \bm{\theta}_{Z,11}- \bm{\theta}_{Z,12} + \bm{\theta}_{Z,12} - \bm{\widehat{\theta}}_{Z,12} \\
&=& (\bm{\widehat{\theta}}_{Z,11} -\bm{\theta}_{Z,11} ) +( \bm{\theta}_{Z,11}- \bm{\theta}_{Z,12}) - (\bm{\widehat{\theta}}_{Z,12} - \bm{\theta}_{Z,12}) \\
&=& (\bm{\widehat{\theta}}_{Z,11} -\bm{\theta}_{Z,11} ) +\beta - (\bm{\widehat{\theta}}_{Z,12} - \bm{\theta}_{Z,12}) \\
&=& \beta + (\bm{\widehat{\theta}}_{Z,11} -\bm{\theta}_{Z,11} ) - (\bm{\widehat{\theta}}_{Z,12} - \bm{\theta}_{Z,12}).
\end{eqnarray} 
Multiplying by $n$ and rearranging terms we finally get 
\begin{eqnarray}
n(\hat{\beta} - \beta) &=& n(\bm{\widehat{\theta}}_{Z,11} -\bm{\theta}_{Z,11} ) - n (\bm{\widehat{\theta}}_{Z,12} - \bm{\theta}_{Z,12}).   
\end{eqnarray}
Hence, understanding the asymptotic behavior of $n(\hat{\beta} - \beta)$ is equivalent to understanding the asymptotic behavior of the difference between $n(\bm{\widehat{\theta}}_{Z,11} -\bm{\theta}_{Z,11} )$ and $n (\bm{\widehat{\theta}}_{Z,12} - \bm{\theta}_{Z,12}) $. It turns out that analogously to what is available for MLE and variational approximations \citep{BickelEtAl2013}, we can prove normality of these two terms for the GRDPG \citep{TangEtAl2017}. \\

\noindent The following Lemma~\ref{lemma:thm2TangEtAl2017} (corresponding to Theorem~2 in \cite{TangEtAl2017}), showing asymptotic normality of the spectral estimator for the SBM probabilities, will be used in the proof. 
\begin{lemma} (Theorem 2 in \cite{TangEtAl2017}) Let $\bm{A}\sim SBM\left( \bm{\theta}, \bm{\eta}\right) $ be a 
$K$-block stochastic blockmodel graph on $n$ vertices. Let $\bm{\mu}_1, \dots, \bm{\mu}_K$ be point masses in 
$\mathbb{R}^d$ such that $\bm{\theta}_{k\ell}=\bm{\mu}_{k}^{T} \bm{I}_{d_1,d_2}\bm{\mu}_\ell$ and let $\Delta=\sum_k \eta_k \bm{\mu}_k \bm{\mu}_k^T$. For $k\in [K]$ and $\ell\in[K]$, let $\psi_{k\ell}$ be\\

\begin{eqnarray}
\psi_{k\ell} &=& \sum_{r=1}^{K}\eta_r \left(\bm{\theta}_{kr}(1-\bm{\theta}_{kr}) + \bm{\theta}_{\ell r}(1-\bm{\theta}_{\ell r})\right) \bm{\mu}_k^T \Delta^{-1} \bm{I}_{d_1,d_2} \Delta^{-1} \bm{\mu}_\ell \\
  &-& \sum_{r=1}^{K}\sum_{s=1}^{K} \eta_r \eta_s
\bm{\theta}_{sr} (1- \bm{\theta}_{sr})  \bm{\mu}_s^T \Delta^{-1} \bm{I}_{d_1,d_2} \Delta^{-1}(\bm{\mu}_\ell \bm{\mu}_k^T + \bm{\mu}_k \bm{\mu}_{\ell}^{T} )\Delta^{-1}\bm{\mu}_s .
\end{eqnarray}

Let $\zeta_{k\ell}=\bm{\mu}_k^T \Delta^{-1}\bm{\mu}_\ell $ and define $\sigma_{kk}^2$ for $k\in[K]$ to be
\begin{eqnarray}
\sigma_{kk}^2  &=& 4 \bm{\theta}_{kk}(1-\bm{\theta}_{kk}) \zeta_{kk}^2 +
4\sum_{r=1}^K \eta_r\bm{\theta}_{kr}(1 - \bm{\theta}_{kr}) \zeta_{kr}^2 \left( \frac{1}{\eta_k} -2 \zeta_{kk}  \right) \\
&+& 2 \sum_{r=1}^K \sum_{s=1}^K  \eta_r \eta_s \bm{\theta}_{rs}(1-\bm{\theta}_{rs}) \zeta_{kr}^2 \zeta_{ks}^2
\end{eqnarray}
and define $\sigma_{k\ell}^2$ for $k\in[K]$ and $\ell\in [K]$, $ k \neq \ell$ to be
\begin{eqnarray}
\sigma_{k\ell}^2 &=& 
\left( \bm{\theta}_{kk}(1-\bm{\theta}_{kk}) + \bm{\theta}_{\ell\ell}(1-\bm{\theta}_{\ell\ell}) \right) \zeta_{kl}^2 + 2 \bm{\theta}_{k\ell}(1-\bm{\theta}_{k\ell}) \zeta_{kk}\zeta_{\ell\ell} \\
&+& \sum_{r=1}^{K} \eta_r \bm{\theta}_{kr}(1-\bm{\theta}_{kr})\zeta_{\ell r}^2 \left(\frac{1}{\eta_k} - 2\zeta_{kk} \right) \\
&+& \sum_{r=1}^{K} \eta_r \bm{\theta}_{\ell r}(1-\bm{\theta}_{\ell r})\zeta_{k r}^2 \left(\frac{1}{\eta_\ell} - 2\zeta_{\ell\ell} \right) \\
&-& 2 \sum_{r=1}^K \eta_r \left( \bm{\theta}_{kr}(1-\bm{\theta}_{kr}) + \bm{\theta}_{\ell r}(1-\bm{\theta}_{\ell r}) \right) \zeta_{kr}\zeta_{r \ell}\zeta_{k\ell} \\
&+& \frac{1}{2}\sum_{r=1}^K \sum_{s=1}^K \eta_r \eta_s \bm{\theta}_{rs}(1- \bm{\theta}_{rs})\left(\zeta_{kr}\zeta_{\ell s} + \zeta_{\ell r}\zeta_{ks}\right)^2 .
\end{eqnarray}
Then for any $k\in[K]$ and $\ell\in [K]$,
\begin{equation}
    n\left( \widehat{\bm{\theta}}_{k\ell} - \bm{\theta}_{k\ell} - \frac{\psi_{k\ell}}{n}\right) \overset{d}{\rightarrow}N(0,\sigma_{k\ell}^2)
\end{equation}
as $n\rightarrow\infty$.
\label{lemma:thm2TangEtAl2017}
\end{lemma}
\begin{proof}
See \cite{TangEtAl2017} for a detailed proof.\\
\end{proof}

Using the result in Lemma \ref{lemma:thm2TangEtAl2017}, we can see that 
\begin{eqnarray}
n(\bm{\widehat{\theta}}_{Z,11} -\bm{\theta}_{Z,11} ) & \overset{d}{\rightarrow} & N( \psi_{11},\sigma_{11}^2) \\
n (\bm{\widehat{\theta}}_{Z,12} - \bm{\theta}_{Z,12}) & \overset{d}{\rightarrow} & N( \psi_{12},\sigma_{12}^2)
\end{eqnarray}
and by consequence
\begin{eqnarray}
n(\hat{\beta} - \beta) &\overset{d}{\rightarrow} & N(\psi_\beta,\sigma_{\beta}^2)   
\end{eqnarray}
where 
\begin{eqnarray}
\psi_\beta  &=& \psi_{11}  -  \psi_{12}\\
\sigma_{\beta}^2 &=& \sigma_{11}^2+\sigma_{12}^2 - 2\sigma_{11,12}
\end{eqnarray}
and we have used notation $\sigma_{k \ell ,k' \ell '}$ to indicated the covariance terms, that is
\begin{eqnarray}
\sigma_{k \ell ,k' \ell '} = \mathbb{COV}(\bm{\widehat{\theta}}_{Z,k\ell} ,\bm{\widehat{\theta}}_{Z,k ' \ell'} )
\end{eqnarray}
for any $k,\ell, k', \ell' \in \lbrace 1, \dots ,K\rbrace$.
The first two terms of the variance are given above in Lemma~\ref{lemma:thm2TangEtAl2017}. The multivariate version of the CLT can be obtained by applying the Cramer-Wold device. For the covariance term $\sigma_{11,12}$, we provide the calculation below.\\

\noindent\textbf{Computation of the covariance terms}\\

We compute the covariance for a slightly more general model, that includes 
a sparsity coefficient $\rho_n$, that is
\begin{equation}
    \bm{P}_{ij} =  \rho_n\left(\bm{X}_i^T\bm{X}_j + \bm{1}_{\lbrace \bm{Z}_i = \bm{Z}_j \rbrace}\right).
\end{equation}
In the main text we separately discuss the cases in which $\rho_n\rightarrow c$, where $c>0$ is a constant, or $\rho_n\rightarrow 0$, but $n\rho_n =\omega(\sqrt{n})$ when  $n\rightarrow\infty$. Let $\bm{s}_k$ be the vector in $\mathbb{R}^n$ whose $i$-th entry is $1$ if $\xi_i=k$ and $0$ otherwise, and let $n_k$ be the number of nodes in block $k$, that is $n_k = \vert \lbrace i: \xi_i=k\rbrace \vert$. \\

We want to compute the correlation between $\hat{\bm{\theta}}_{Z,k\ell}$ and $\hat{\bm{\theta}}_{Z,k'\ell'}$, for $k\neq k'$ and $\ell\neq \ell '$ in general. \\

To simplify notation, we will omit the $Z$ from the subscript, so we will refer to  $\hat{\bm{\theta}}_{k\ell}$ instead of $\hat{\bm{\theta}}_{Z,k\ell}$ for any $k,\ell$.
Let $\bm{S}$ be the $d\times d$ diagonal matrix containing the largest $d$ eigenvalues of $\bm{P}$ in absolute value, and let $\bm{U}$ be the $n\times d$ matrix whose rows are the corresponding eigenvectors of $\bm{P}$. \\

We start from equation (A.5) in the appendix of \cite{TangEtAl2017}.
\begin{eqnarray}
\frac{n\rho_n^{1/2}}{n_k n_\ell}\left(\hat{\bm{\theta}}_{k\ell}- \bm{\theta}_{k\ell}\right) &=& 
\frac{n\rho_n^{-1/2}}{n_k n_\ell}
\left( \bm{s}_k^T \bm{E} \bm{\Pi}_U \bm{s}_\ell + \bm{s}_{\ell}^T  \bm{\Pi}_U^{\bot} \bm{E} \bm{\Pi}_U\bm{s}_k \right) \label{eq:mainA.5} \\
 &+& \frac{n\rho_n^{-1/2}}{n_k n_\ell}
\left( \bm{s}_k^T \bm{\Pi}_U^\bot \bm{E}^2 \bm{P}^\dagger \bm{s}_\ell + \bm{s}_{\ell}^T  \bm{\Pi}_U^{\bot} \bm{E}^2 \bm{P}^\dagger \bm{s}_k \right) \label{eq:biasA.5} \\
&+& O_{p}\left(n^{-1/2} \rho_n^{-1}\right) \label{eq:OpA.5}
\end{eqnarray}
where $\bm{E}=\bm{A}-\bm{P}$, $\bm{\Pi}_U= \bm{U}\bm{U}^T$, $\bm{\Pi}_U^{\bot}=\bm{I}-\bm{\Pi}_U$ and $\bm{P}^\dagger = \bm{U}\bm{S}^{-1}\bm{U}^T$. \\

Term (\ref{eq:OpA.5}) goes to zero when $n\rho_{n}=\omega(\sqrt{n})$ as $n\rightarrow\infty$. \\

Term (\ref{eq:biasA.5}) is the bias term, corresponding to $\psi_{k\ell}$ or $\rho_{n}^{-1/2}\tilde{\psi}_{k\ell}$ depending on whether $\rho_n \equiv 1$ or $\rho_n\rightarrow 0$, respectively. \\

Term (\ref{eq:mainA.5}) is the leading order term which converges in distribution to a normal random variable. In particular, this is the term that must be considered when deriving asymptotic covariances.

Towards this end, define $\bm{\Upsilon}_{k\ell}$ below as in equation (A.6) found in \cite{TangEtAl2017}, namely
\begin{eqnarray}
\bm{\Upsilon}_{k\ell} & := & \frac{n\rho_n^{-1/2}}{n_k n_\ell}
\left( \bm{s}_k^T \bm{E} \bm{\Pi}_U \bm{s}_\ell + \bm{s}_{\ell}^T  \bm{\Pi}_U^{\bot} \bm{E} \bm{\Pi}_U \bm{s}_k \right) \\
&=& \frac{n\rho_n^{-1/2}}{n_k n_\ell}
 tr \bm{E} \left(   \bm{\Pi}_U \bm{s}_\ell \bm{s}_k^T +   \bm{\Pi}_U \bm{s}_k \bm{s}_{\ell}^T \bm{\Pi}_U^{\bot} \right) \\
 &=& \frac{n\rho_n^{-1/2}}{n_k n_\ell}
 tr \bm{E} \left(   \bm{\Pi}_U \bm{s}_\ell \bm{s}_k^T +   \bm{\Pi}_U \bm{s}_k \bm{s}_{\ell}^T  - \bm{\Pi}_U \bm{s}_k \bm{s}_{\ell}^T \bm{\Pi}_U \right) \\
 &=& \frac{n\rho_n^{-1/2}}{n_k n_\ell}
 tr\left(\bm{A} - \bm{P} \right)\left(   \bm{\Pi}_U \bm{s}_\ell \bm{s}_k^T +   \bm{\Pi}_U \bm{s}_k \bm{s}_{\ell}^T  - \bm{\Pi}_U \bm{s}_k \bm{s}_{\ell}^T \bm{\Pi}_U \right) \\
 &=& \frac{n\rho_n^{-1/2}}{n_k n_\ell}
 tr\left(\bm{A} - \bm{P} \right) \bm{M}
\end{eqnarray}
where $\bm{M} :=  \bm{\Pi}_U \bm{s}_\ell \bm{s}_k^T +   \bm{\Pi}_U \bm{s}_k \bm{s}_{\ell}^T  - \bm{\Pi}_U \bm{s}_k \bm{s}_{\ell}^T \bm{\Pi}_U$. 

We therefore have:
\begin{eqnarray}
\bm{\Upsilon}_{k\ell} = \frac{n\rho_n^{-1/2}}{n_k n_\ell} 
\sum_i \sum_j \left( \bm{A}_{ij} - \bm{P}_{ij}\right) \bm{M}_{ij}.
\end{eqnarray}

First, note that the variable $\bm{\Upsilon}_{k\ell}$ has expected value equal to zero, i.e., $\mathbb{E}[\bm{\Upsilon}_{k \ell}] = 0$, since $\mathbb{E}[\bm{A}_{ij}] = \bm{P}_{ij}$ for each pair $\{i,j\}$. This implies that we can focus on computing covariances as
\begin{equation}
    \mathbb{COV}[\bm{\Upsilon}_{k\ell},\bm{\Upsilon}_{k' \ell'} ] = \mathbb{E}[\bm{\Upsilon}_{k\ell}\bm{\Upsilon}_{k' \ell'} ].
\end{equation}

Towards this end, let the matrix $\bm{Q}$ be defined as
\begin{equation}
    \bm{Q} =   \bm{\Pi}_U \bm{s}_{\ell'} \bm{s}_{k'}^T +   \bm{\Pi}_U \bm{s}_{k'} \bm{s}_{\ell'}^T  - \bm{\Pi}_U \bm{s}_{k'} \bm{s}_{\ell'}^T \bm{\Pi}_U ,
\end{equation}

and  so
\begin{eqnarray}
\bm{\Upsilon}_{k' \ell'} = \frac{n\rho_n^{-1/2}}{n_{k'} n_{\ell'}} 
\sum_i \sum_j \left( \bm{A}_{ij} - \bm{P}_{ij}\right) \bm{Q}_{ij}
\end{eqnarray}

The product $\bm{\Upsilon}_{k\ell}\bm{\Upsilon}_{k' \ell'} $ is then given by
\begin{eqnarray}
\bm{\Upsilon}_{k\ell}\bm{\Upsilon}_{k' \ell'} &=& \left(\frac{n\rho_n^{-1/2}}{n_k n_\ell} 
\sum_i \sum_j \left( \bm{A}_{ij} - \bm{P}_{ij}\right) \bm{M}_{ij} \right)\left( \frac{n\rho_n^{-1/2}}{n_{k'} n_{\ell'}} 
\sum_{i'} \sum_{j'} \left( \bm{A}_{i' j'} - \bm{P}_{i' j'}\right) \bm{Q}_{i' j'} \right) \\
&=& \frac{n^2 \rho_n^{-1}}{n_k n_\ell n_{k'} n_{\ell'}} \sum_i \sum_j \sum_{i'} \sum_{j'} \left( \bm{A}_{ij} - \bm{P}_{ij}\right) \left( \bm{A}_{i' j'} - \bm{P}_{i' j'}\right)
\bm{M}_{ij}\bm{Q}_{i' j'}
\end{eqnarray}

First note that when $\{i,j\}\neq \{i',j'\}$, then the expected value of the corresponding term in the above summation is zero, since then $(\bm{A}_{ij}-\bm{P}_{ij})$ and $(\bm{A}_{i' j'}-\bm{P}_{i' j'})$ are independent, centered random variables. Therefore we can focus on the case when $\{i,j\}=\{i',j'\}$.

Further expanding $\mathbb{E}[\bm{\Upsilon}_{k\ell}\bm{\Upsilon}_{k' \ell'}]$ subsequently yields

\begin{eqnarray}
\mathbb{E}[\bm{\Upsilon}_{k\ell}\bm{\Upsilon}_{k' \ell'}] &=& \frac{n^2 \rho_n^{-1}}{n_k n_\ell n_{k'} n_{\ell'}} \sum_i \sum_j \sum_{i'} \sum_{j'} \mathbb{E}\left[ \left( \bm{A}_{ij} - \bm{P}_{ij}\right) \left( \bm{A}_{i' j'} - \bm{P}_{i' j'}\right) \right]
\bm{M}_{ij}\bm{Q}_{i' j'} \\
&=& \frac{n^2 \rho_n^{-1}}{n_k n_\ell n_{k'} n_{\ell'}} \sum_i \sum_j  \mathbb{E}\left[ \left( \bm{A}_{ij} - \bm{P}_{ij}\right)^2  \right]
\bm{M}_{ij}\bm{Q}_{i j} \\
&=& \frac{n^2 \rho_n^{-1}}{n_k n_\ell n_{k'} n_{\ell'}} \sum_i \sum_j   \bm{P}_{ij}\left(1-\bm{P}_{ij}\right)
\bm{M}_{ij}\bm{Q}_{i j} .
\end{eqnarray}

We thus need to compute the entries of $\bm{M}$ and $\bm{Q}$ to obtain a computable formula.

We have for $\bm{M}$ that $\bm{M}_{ij}:= \upsilon_{ij}^{(1)} + \upsilon_{ij}^{(2)} - \upsilon_{ij}^{(3)}$, where
\begin{eqnarray}
\upsilon_{ij}^{(1)} := \left(\bm{\Pi}_U \bm{s}_\ell \bm{s}_k^T \right)_{ij}
    &=&
    n_\ell \bm{Y}_i^T (\bm{Y}^T \bm{Y})^{-1} \bm{\mu}_\ell \mathds{1}\lbrace \xi_j = k \rbrace \\
\upsilon_{ij}^{(2)} := \left(\bm{\Pi}_U \bm{s}_k \bm{s}_{\ell}^T \right)_{ij}
    &=&
    n_k \bm{Y}_i^T (\bm{Y}^T \bm{Y})^{-1} \bm{\mu}_k \mathds{1}\lbrace \xi_j = \ell \rbrace \\
\upsilon_{ij}^{(3)} := \left(\bm{\Pi}_U \bm{s}_k \bm{s}_{\ell}^T \bm{\Pi}_U \right)_{ij}
    &=&
    n_k n_\ell \bm{Y}_i^T (\bm{Y}^T \bm{Y})^{-1} \bm{\mu}_k \bm{\mu}_{\ell}^T  (\bm{Y}^T \bm{Y})^{-1}\bm{Y}_j\\
\
\end{eqnarray}
and analogously for $\bm{Q}_{ij} := \varrho_{ij}^{(1)} + \varrho_{ij}^{(2)} - \varrho_{ij}^{(3)}$. Hence, $\bm{M}_{ij}\bm{Q}_{ij} =(\upsilon_{ij}^{(1)} + \upsilon_{ij}^{(2)} - \upsilon_{ij}^{(3)})(\varrho_{ij}^{(1)} + \varrho_{ij}^{(2)} - \varrho_{ij}^{(3)})$. In what follows, it will sometimes be useful to write the scalars $\upsilon_{ij}^{(\alpha)}, \varrho_{ij}^{(\beta)}$, for $\alpha,\beta\in\lbrace 1,2,3\rbrace$, equivalently in terms of their transpose (i.e.,~see the right-hand sides of their definitions).

To begin, we make a preliminary observation that
\begin{align*}
    &\tfrac{\rho_{n}^{-1}}{n_{k}n_{k'}}\left(\sum_{i,j}\bm{P}_{ij}(1-\bm{P}_{ij})\bm{Y}_{i}\bm{Y}_{i}^{T}\mathds{1}\{\xi_{j}=k\}\mathds{1}\{\xi_{j}=k'\}\right)\\
    &\hspace{3em}\overset{a.s.}{\rightarrow} \begin{cases}
        \tfrac{1}{\eta_{k}}\mathbb{E}\left[\bm{\theta}_{k\xi(\bm{Y}_{1})}(1-\bm{\theta}_{k\xi(\bm{Y}_{1})})\bm{Y}_{1}\bm{Y}_{1}^{T}\right] \hfill\hspace{1em} \textnormal{ if $\rho_{n}\equiv 1$ and $k=k'$} ,\\
        \tfrac{1}{\eta_{k}}\mathbb{E}\left[\bm{\theta}_{k\xi(\bm{Y}_{1})}\bm{Y}_{1}\bm{Y}_{1}^{T}\right] \hfill\hspace{1em} \textnormal{ if $\rho_{n} \rightarrow 0$ and $k=k'$} ,\\
        0 \hfill\hspace{1em} \textnormal{ if $k\neq k'$} .
    \end{cases}
\end{align*}

Consider the terms involving $\upsilon_{ij}^{(1)}\varrho_{ij}^{(1)}$, $\upsilon_{ij}^{(1)}\varrho_{ij}^{(2)}$,
$\upsilon_{ij}^{(2)}\varrho_{ij}^{(1)}$, and
$\upsilon_{ij}^{(2)}\varrho_{ij}^{(2)}$. Then

\begin{align*}
    &\tfrac{n^2 \rho_n^{-1}}{n_k n_\ell n_{k'} n_{\ell'}} \sum_{i,j}   \bm{P}_{ij}\left(1-\bm{P}_{ij}\right) \upsilon_{ij}^{(1)}\varrho_{ij}^{(1)}\\
    &\hspace{1em}= \tfrac{n^2 \rho_n^{-1}}{n_k n_\ell n_{k'} n_{\ell'}}\left( n_{\ell}n_{\ell'}\bm{\mu}_{\ell}^{T}(\bm{Y}^{T}\bm{Y})^{-1}\left(\sum_{i,j}\bm{P}_{ij}(1-\bm{P}_{ij})\bm{Y}_{i}\bm{Y}_{i}^{T}\mathds{1}\{\xi_{j}=k\}\mathds{1}\{\xi_{j} = k'\}\right)(\bm{Y}^{T}\bm{Y})^{-1}\bm{\mu}_{\ell'}
    \right)\\
    &\hspace{1em}\overset{a.s.}{\rightarrow} \begin{cases}
        \tfrac{1}{\eta_{k}}\bm{\mu}_{\ell}^{T}\bm{\Delta}^{-1}\mathbb{E}\left[\bm{\theta}_{k\xi(\bm{Y}_{1})}(1-\bm{\theta}_{k\xi(\bm{Y}_{1})})\bm{Y}_{1}\bm{Y}_{1}^{T}\right]\bm{\Delta}^{-1}\bm{\mu}_{\ell'} \hfill\hspace{1em} \textnormal{ if $\rho_{n}\equiv 1$ and $k=k'$, }\\
        \tfrac{1}{\eta_{k}}\bm{\mu}_{\ell}^{T}\bm{\Delta}^{-1}\mathbb{E}\left[\bm{\theta}_{k\xi(\bm{Y}_{1})}\bm{Y}_{1}\bm{Y}_{1}^{T}\right]\bm{\Delta}^{-1}\bm{\mu}_{\ell'} \hfill\hspace{1em} \textnormal{ if $\rho_{n} \rightarrow 0$ and $k=k'$, }\\
        0 \hfill\hspace{1em} \textnormal{ if $k\neq k'$, }
\end{cases}
\end{align*}

and similarly

\begin{align*}
&\tfrac{n^2 \rho_n^{-1}}{n_k n_\ell n_{k'} n_{\ell'}} \sum_{i,j}   \bm{P}_{ij}\left(1-\bm{P}_{ij}\right) \upsilon_{ij}^{(1)}\varrho_{ij}^{(2)}\\
    &\hspace{1em}\overset{a.s.}{\rightarrow} \begin{cases}
        \tfrac{1}{\eta_{k}}\bm{\mu}_{\ell}^{T}\bm{\Delta}^{-1}\mathbb{E}\left[\bm{\theta}_{k\xi(\bm{Y}_{1})}(1-\bm{\theta}_{k\xi(\bm{Y}_{1})})\bm{Y}_{1}\bm{Y}_{1}^{T}\right]\bm{\Delta}^{-1}\bm{\mu}_{k'} \hfill\hspace{1em} \textnormal{ if $\rho_{n}\equiv 1$ and $k=\ell'$, }\\
        \tfrac{1}{\eta_{k}}\bm{\mu}_{\ell}^{T}\bm{\Delta}^{-1}\mathbb{E}\left[\bm{\theta}_{k\xi(\bm{Y}_{1})}\bm{Y}_{1}\bm{Y}_{1}^{T}\right]\bm{\Delta}^{-1}\bm{\mu}_{k'} \hfill\hspace{1em} \textnormal{ if $\rho_{n} \rightarrow 0$ and $k=\ell'$, }\\
        0 \hfill\hspace{1em} \textnormal{ if $k\neq \ell'$, }
\end{cases}
\end{align*}

and similarly

\begin{align*}
    &\tfrac{n^2 \rho_n^{-1}}{n_k n_\ell n_{k'} n_{\ell'}} \sum_{i,j}   \bm{P}_{ij}\left(1-\bm{P}_{ij}\right) \upsilon_{ij}^{(2)}\varrho_{ij}^{(1)}\\
    &\hspace{1em}\overset{a.s.}{\rightarrow} \begin{cases}
        \tfrac{1}{\eta_{\ell}}\bm{\mu}_{k}^{T}\bm{\Delta}^{-1}\mathbb{E}\left[\bm{\theta}_{\ell\xi(\bm{Y}_{1})}(1-\bm{\theta}_{\ell\xi(\bm{Y}_{1})})\bm{Y}_{1}\bm{Y}_{1}^{T}\right]\bm{\Delta}^{-1}\bm{\mu}_{\ell'} \hfill\hspace{1em} \textnormal{ if $\rho_{n}\equiv 1$ and $\ell=k'$, }\\
        \tfrac{1}{\eta_{\ell}}\bm{\mu}_{k}^{T}\bm{\Delta}^{-1}\mathbb{E}\left[\bm{\theta}_{\ell\xi(\bm{Y}_{1})}\bm{Y}_{1}\bm{Y}_{1}^{T}\right]\bm{\Delta}^{-1}\bm{\mu}_{\ell'} \hfill\hspace{1em} \textnormal{ if $\rho_{n} \rightarrow 0$ and $\ell=k'$, }\\
        0 \hfill\hspace{1em} \textnormal{ if $\ell\neq k'$. }
\end{cases}
\end{align*}

and similarly

\begin{align*}
&\tfrac{n^2 \rho_n^{-1}}{n_k n_\ell n_{k'} n_{\ell'}} \sum_{i,j}   \bm{P}_{ij}\left(1-\bm{P}_{ij}\right) \upsilon_{ij}^{(2)}\varrho_{ij}^{(2)}\\
    &\hspace{1em}\overset{a.s.}{\rightarrow} \begin{cases}
        \tfrac{1}{\eta_{\ell}}\bm{\mu}_{k}^{T}\bm{\Delta}^{-1}\mathbb{E}\left[\bm{\theta}_{\ell\xi(\bm{Y}_{1})}(1-\bm{\theta}_{\ell\xi(\bm{Y}_{1})})\bm{Y}_{1}\bm{Y}_{1}^{T}\right]\bm{\Delta}^{-1}\bm{\mu}_{k'} \hfill\hspace{1em} \textnormal{ if $\rho_{n}\equiv 1$ and $\ell=\ell'$, }\\
        \tfrac{1}{\eta_{\ell}}\bm{\mu}_{k}^{T}\bm{\Delta}^{-1}\mathbb{E}\left[\bm{\theta}_{\ell\xi(\bm{Y}_{1})}\bm{Y}_{1}\bm{Y}_{1}^{T}\right]\bm{\Delta}^{-1}\bm{\mu}_{k'} \hfill\hspace{1em} \textnormal{ if $\rho_{n} \rightarrow 0$ and $\ell=\ell'$, }\\
        0 \hfill\hspace{1em} \textnormal{ if $\ell\neq \ell'$. }
\end{cases}
\end{align*}

Next, we consider the terms involving $\upsilon_{ij}^{(1)}\varrho_{ij}^{(3)}$,
$\upsilon_{ij}^{(2)}\varrho_{ij}^{(3)}$,
$\upsilon_{ij}^{(3)}\varrho_{ij}^{(1)}$, and 
$\upsilon_{ij}^{(3)}\varrho_{ij}^{(2)}$. In particular,

\begin{align*}
    & -\tfrac{n^2 \rho_n^{-1}}{n_k n_\ell n_{k'} n_{\ell'}} \sum_{i,j}   \bm{P}_{ij}\left(1-\bm{P}_{ij}\right) \upsilon_{ij}^{(1)}\varrho_{ij}^{(3)}\\
    &\hspace{1em}= -
        \tfrac{n^2 \rho_n^{-1}}{n_k n_\ell n_{k'} n_{\ell'}} \sum_{i,j}   \bm{P}_{ij}\left(1-\bm{P}_{ij}\right)\left(n_{\ell}n_{k'}n_{\ell'}\mathds{1}\{\xi_{j}=k\}\bm{\mu}_{\ell}^{T}(\bm{Y}^{T}\bm{Y})^{-1}\bm{Y}_{i}\bm{Y}_{i}^{T}(\bm{Y}^{T}\bm{Y})^{-1}\bm{\mu}_{k'}\bm{\mu}_{\ell'}^{T}(\bm{Y}^{T}\bm{Y})^{-1}\bm{Y}_{j}\right)\\
        &\hspace{1em}\overset{a.s.}{\rightarrow} \begin{cases}
           - \bm{\mu}_{\ell}^{T}\bm{\Delta}^{-1}\mathbb{E}\left[\bm{\theta}_{k\xi(\bm{Y}_{1})}(1-\bm{\theta}_{k\xi(\bm{Y}_{1})})\bm{Y}_{1}\bm{Y}_{1}^{T}\right]\bm{\Delta}^{-1}\bm{\mu}_{k'}\bm{\mu}_{\ell'}^{T}\bm{\Delta}^{-1}\bm{\mu}_{k} \hfill\hspace{1em} \textnormal{ if $\rho_{n}\equiv 1$, }\\
           - \bm{\mu}_{\ell}^{T}\bm{\Delta}^{-1}\mathbb{E}\left[\bm{\theta}_{k\xi(\bm{Y}_{1})}\bm{Y}_{1}\bm{Y}_{1}^{T}\right]\bm{\Delta}^{-1}\bm{\mu}_{k'}\bm{\mu}_{\ell'}^{T}\bm{\Delta}^{-1}\bm{\mu}_{k} \hfill\hspace{1em} \textnormal{ if $\rho_{n} \rightarrow 0$, }            
        \end{cases}
\end{align*}

and similarly
\begin{align*}
    &- \tfrac{n^2 \rho_n^{-1}}{n_k n_\ell n_{k'} n_{\ell'}} \sum_{i,j}   \bm{P}_{ij}\left(1-\bm{P}_{ij}\right) \upsilon_{ij}^{(2)}\varrho_{ij}^{(3)}\\
        &\hspace{1em}\overset{a.s.}{\rightarrow} \begin{cases}
           - \bm{\mu}_{k}^{T}\bm{\Delta}^{-1}\mathbb{E}\left[\bm{\theta}_{\ell\xi(\bm{Y}_{1})}(1-\bm{\theta}_{\ell\xi(\bm{Y}_{1})})\bm{Y}_{1}\bm{Y}_{1}^{T}\right]\bm{\Delta}^{-1}\bm{\mu}_{k'}\bm{\mu}_{\ell'}^{T}\bm{\Delta}^{-1}\bm{\mu}_{\ell} \hfill\hspace{1em} \textnormal{ if $\rho_{n}\equiv 1$, }\\
            - \bm{\mu}_{k}^{T}\bm{\Delta}^{-1}\mathbb{E}\left[\bm{\theta}_{\ell\xi(\bm{Y}_{1})}\bm{Y}_{1}\bm{Y}_{1}^{T}\right]\bm{\Delta}^{-1}\bm{\mu}_{k'}\bm{\mu}_{\ell'}^{T}\bm{\Delta}^{-1}\bm{\mu}_{\ell} \hfill\hspace{1em} \textnormal{ if $\rho_{n} \rightarrow 0$. }            
        \end{cases}
\end{align*}

Along the same lines,

\begin{align*}
    &- \tfrac{n^2 \rho_n^{-1}}{n_k n_\ell n_{k'} n_{\ell'}} \sum_{i,j}   \bm{P}_{ij}\left(1-\bm{P}_{ij}\right) \upsilon_{ij}^{(3)}\varrho_{ij}^{(1)}\\
        &\hspace{1em}= 
            -\tfrac{n^2 \rho_n^{-1}}{n_k n_\ell n_{k'} n_{\ell'}} \sum_{i,j}   \bm{P}_{ij}\left(1-\bm{P}_{ij}\right)\left(n_{k}n_{\ell}n_{\ell'}\mathds{1}\{\xi_{j}=k'\}\bm{Y}_{j}^{T}(\bm{Y}^{T}\bm{Y})^{-1}\bm{\mu}_{\ell}\bm{\mu}_{k}^{T}(\bm{Y}^{T}\bm{Y})^{-1}\bm{Y}_{i}\bm{Y}_{i}^{T}(\bm{Y}^{T}\bm{Y})^{-1}\bm{\mu}_{\ell'}\right)\\
        &\hspace{1em}\overset{a.s.}{\rightarrow} \begin{cases}
            -\bm{\mu}_{k'}^{T}\bm{\Delta}^{-1}\bm{\mu}_{\ell}\bm{\mu}_{k}^{T}\bm{\Delta}^{-1}\mathbb{E}\left[\bm{\theta}_{k'\xi(\bm{Y}_{1})}(1-\bm{\theta}_{k'\xi(\bm{Y}_{1})})\bm{Y}_{1}\bm{Y}_{1}^{T}\right]\bm{\Delta}^{-1}\bm{\mu}_{\ell'} \hfill\hspace{1em} \textnormal{ if $\rho_{n}\equiv 1$, }\\
            -\bm{\mu}_{k'}^{T}\bm{\Delta}^{-1}\bm{\mu}_{\ell}\bm{\mu}_{k}^{T}\bm{\Delta}^{-1}\mathbb{E}\left[\bm{\theta}_{k'\xi(\bm{Y}_{1})}\bm{Y}_{1}\bm{Y}_{1}^{T}\right]\bm{\Delta}^{-1}\bm{\mu}_{\ell'} \hfill\hspace{1em} \textnormal{ if $\rho_{n} \rightarrow 0$, }            
\end{cases}
\end{align*}

and similarly

\begin{align*}
    &-\tfrac{n^2 \rho_n^{-1}}{n_k n_\ell n_{k'} n_{\ell'}} \sum_{i,j}   \bm{P}_{ij}\left(1-\bm{P}_{ij}\right) \upsilon_{ij}^{(3)}\varrho_{ij}^{(2)}\\
        &\hspace{1em}\overset{a.s.}{\rightarrow} \begin{cases}
            -\bm{\mu}_{\ell'}^{T}\bm{\Delta}^{-1}\bm{\mu}_{\ell}\bm{\mu}_{k}^{T}\bm{\Delta}^{-1}\mathbb{E}\left[\bm{\theta}_{\ell'\xi(\bm{Y}_{1})}(1-\bm{\theta}_{\ell'\xi(\bm{Y}_{1})})\bm{Y}_{1}\bm{Y}_{1}^{T}\right]\bm{\Delta}^{-1}\bm{\mu}_{k'} \hfill\hspace{1em} \textnormal{ if $\rho_{n}\equiv 1$, }\\
            -\bm{\mu}_{\ell'}^{T}\bm{\Delta}^{-1}\bm{\mu}_{\ell}\bm{\mu}_{k}^{T}\bm{\Delta}^{-1}\mathbb{E}\left[\bm{\theta}_{\ell'\xi(\bm{Y}_{1})}\bm{Y}_{1}\bm{Y}_{1}^{T}\right]\bm{\Delta}^{-1}\bm{\mu}_{k'} \hfill\hspace{1em} \textnormal{ if $\rho_{n} \rightarrow 0$. }            
\end{cases}
\end{align*}

Finally, consider the term involving $\upsilon_{ij}^{(3)}\varrho_{ij}^{(3)}$. We see that

\begin{align*}
    &\tfrac{n^2 \rho_n^{-1}}{n_k n_\ell n_{k'} n_{\ell'}} \sum_{i,j}   \bm{P}_{ij}\left(1-\bm{P}_{ij}\right) \upsilon_{ij}^{(3)}\varrho_{ij}^{(3)}\\
        &\hspace{-1em}= 
            \tfrac{n^2 \rho_n^{-1}}{n_k n_\ell n_{k'} n_{\ell'}} \sum_{i,j}   \bm{P}_{ij}\left(1-\bm{P}_{ij}\right)\left(n_{k}n_{\ell}n_{k'}n_{\ell'}\bm{Y}_{j}^{T}(\bm{Y}^{T}\bm{Y})^{-1}\bm{\mu}_{\ell}\bm{\mu}_{k}^{T}(\bm{Y}^{T}\bm{Y})^{-1}\bm{Y}_{i}\bm{Y}_{i}^{T}(\bm{Y}^{T}\bm{Y})^{-1}\bm{\mu}_{k'}\bm{\mu}_{\ell'}^{T}(\bm{Y}^{T}\bm{Y})^{-1}\bm{Y}_{j}\right).
\end{align*}
In order to analyze this quantity, we now decompose the sum over the index $j$ using the indicator variables $\mathds{1}\{\xi_{j}=\alpha\}$ for all SBM blocks $\alpha \in \mathcal{A}$. For each such term we get
\begin{align*}
    &\tfrac{n^2 \rho_n^{-1}}{n_k n_\ell n_{k'} n_{\ell'}} \sum_{i,j}   \bm{P}_{ij}\left(1-\bm{P}_{ij}\right) \upsilon_{ij}^{(3)}\varrho_{ij}^{(3)}\mathds{1}\{\xi_{j}=\alpha\}\\
        &\hspace{1em}\overset{a.s.}{\rightarrow}
        \begin{cases}
            \eta_{\alpha} \bm{\mu}_{\alpha}^{T}\bm{\Delta}^{-1}\bm{\mu}_{\ell}\bm{\mu}_{k}^{T}\bm{\Delta}^{-1}\mathbb{E}\left[\bm{\theta}_{\alpha\xi(\bm{Y}_{1})}(1-\bm{\theta}_{\alpha\xi(\bm{Y}_{1})})\bm{Y}_{1}\bm{Y}_{1}^{T}\right]\bm{\Delta}^{-1}\bm{\mu}_{k'}\bm{\mu}_{\ell'}^{T}\bm{\Delta}^{-1}\bm{\mu}_{\alpha} \hfill\hspace{1em} \textnormal{ if $\rho_{n} \equiv 1$, }\\
            \eta_{\alpha} \bm{\mu}_{\alpha}^{T}\bm{\Delta}^{-1}\bm{\mu}_{\ell}\bm{\mu}_{k}^{T}\bm{\Delta}^{-1}\mathbb{E}\left[\bm{\theta}_{\alpha\xi(\bm{Y}_{1})}\bm{Y}_{1}\bm{Y}_{1}^{T}\right]\bm{\Delta}^{-1}\bm{\mu}_{k'}\bm{\mu}_{\ell'}^{T}\bm{\Delta}^{-1}\bm{\mu}_{\alpha}
            \hfill\hspace{1em} \textnormal{ if $\rho_{n} \rightarrow 0$. }
        \end{cases}
\end{align*}
Hence, by aggregating over all $\alpha\in\mathcal{A}$, we obtain
\begin{align*}
    &\tfrac{n^2 \rho_n^{-1}}{n_k n_\ell n_{k'} n_{\ell'}} \sum_{i,j}   \bm{P}_{ij}\left(1-\bm{P}_{ij}\right) \upsilon_{ij}^{(3)}\varrho_{ij}^{(3)}\\
    &\hspace{1em}\overset{a.s.}{\rightarrow}
        \begin{cases}
            \sum_{\alpha} \eta_{\alpha} \bm{\mu}_{\alpha}^{T}\bm{\Delta}^{-1}\bm{\mu}_{\ell}\bm{\mu}_{k}^{T}\bm{\Delta}^{-1}\mathbb{E}\left[\bm{\theta}_{\alpha\xi(\bm{Y}_{1})}(1-\bm{\theta}_{\alpha\xi(\bm{Y}_{1})})\bm{Y}_{1}\bm{Y}_{1}^{T}\right]\bm{\Delta}^{-1}\bm{\mu}_{k'}\bm{\mu}_{\ell'}^{T}\bm{\Delta}^{-1}\bm{\mu}_{\alpha} \hfill\hspace{1em} \textnormal{ if $\rho_{n} \equiv 1$, }\\
            \sum_{\alpha} \eta_{\alpha} \bm{\mu}_{\alpha}^{T}\bm{\Delta}^{-1}\bm{\mu}_{\ell}\bm{\mu}_{k}^{T}\bm{\Delta}^{-1}\mathbb{E}\left[\bm{\theta}_{\alpha\xi(\bm{Y}_{1})}\bm{Y}_{1}\bm{Y}_{1}^{T}\right]\bm{\Delta}^{-1}\bm{\mu}_{k'}\bm{\mu}_{\ell'}^{T}\bm{\Delta}^{-1}\bm{\mu}_{\alpha}
            \hfill\hspace{1em} \textnormal{ if $\rho_{n} \rightarrow 0$. }
        \end{cases}
\end{align*}
Combining all of the above observations yields $\mathbb{COV}[\bm{\Upsilon}_{k\ell}\bm{\Upsilon}_{k'\ell'}]$ for all possible relationships between tuples $\{k,\ell\}, \{k',\ell'\}$ and for both regimes $\rho_{n}\equiv 1, \rho_{n}\rightarrow 0$. \\


\end{proof}

\subsection{Blocks known and general link function}
If the link function is not the identity, then our general model becomes
\begin{equation}
\bm{A}_{ij}\vert \bm{X}_i, \bm{X}_j, \bm{Z}_i, \bm{Z}_j,\beta \overset{ind}{\sim} Bernoulli \left( h( \bm{X}_i^T  \bm{X}_j + \beta \mathbf{1}_{\lbrace \bm{Z}_i = \bm{Z}_j \rbrace} )\right).
\label{eq:model_known_tau_h}
\end{equation}
A popular choice of $h$ is the logistic specification, with $h(u) = e^u /(1+e^u)$ \citep{ChoiEtAl2011,TraudEtAl2012,Sweet2015,Nimczik2018}. The generalization of the central limit theorem is as follows.

\begin{theorem} (General nonlinear $h$ function) Let $\bm{A}$ be an adjacency matrix from model (\ref{eq:model_known_tau_h}) and let $h$ be a link function, $h:\mathcal{\bm{X}}\times \mathcal{\bm{X}}\times \mathcal{\bm{Z}}\times \mathcal{\bm{Z}}\rightarrow [0,1]$. Let $\bm{\tau}$ be known and let function $g$ be defined as the inverse of $h$, that is $g(\cdot)=h^{-1}(\cdot)$, with first derivative $g^{\prime}(\cdot)$. Let $g^{\prime}(\bm{\nu}_{1}^T \bm{\nu}_{1} +\beta)\neq 0$ and $g^{\prime}(\bm{\nu}_{1}^T \bm{\nu}_{2} )\neq 0$. Then $\widehat{\beta} = h^{-1}(\widehat{\bm{\theta}}_{Z,11})-h^{-1}(\widehat{\bm{\theta}}_{Z,12}) $ is asymptotically normal, in particular 
\begin{eqnarray}
n \left(\hat{\beta} - \beta - \frac{\widetilde{\psi}_{\beta}}{n} \right)\overset{d}{\longrightarrow} N\left(0,\widetilde{\sigma}_{\beta}^2 \right)
\end{eqnarray}

where $\widetilde{\psi}_\beta$ and $\widetilde{\sigma}_\beta^2$  are computed in appendix.
\label{thm:clt_general_h_block_known}
\end{theorem}
\begin{proof}

In this case model (\ref{eq:general_model_sbm}) is as follows:
\begin{equation}
\bm{A}_{ij}\vert \bm{X}_i, \bm{X}_j, \bm{Z}_i, \bm{Z}_j,\beta \overset{ind}{\sim} Bernoulli \left( h\left( \bm{X}_i^T  \bm{X}_j + \beta \mathbf{1}_{\lbrace \bm{Z}_i = \bm{Z}_j \rbrace} \right) \right).
\end{equation}

All the rest is the same, and we can write our SBM as $ (\bm{A}, \bm{\xi}, \bm{Z}) \sim SBM(\bm{\theta}_Z,\bm{\eta})$ with $\widetilde{K}\times \widetilde{K}$ matrix of probabilities $\bm{\theta}_Z$  given by
\begin{tiny}
\begin{equation}
\bm{\theta}_Z = \bordermatrix{ & \tau=1;Z=0 & \tau=1;Z=1 & \tau=2;Z=0 & \tau=2;Z=1 & \cdots & \tau=K;Z=0 & \tau=K;Z=1 \cr
                \tau=1;Z=0 & h(\bm{\nu}_1^T  \bm{\nu}_1 + \beta) & h(\bm{\nu}_1^T  \bm{\nu}_1)  & h(\bm{\nu}_1^T  \bm{\nu}_2  + \beta )& h(\bm{\nu}_1^T \bm{\nu}_2) & \cdots & h(\bm{\nu}_1^T  \bm{\nu}_K  + \beta) & h(\bm{\nu}_1^T  \bm{\nu}_K) \cr                
                \tau=1;Z=1 & h(\bm{\nu}_1^T  \bm{\nu}_1) & h(\bm{\nu}_1^T  \bm{\nu}_1 + \beta) & h(\bm{\nu}_1^T \bm{\nu}_2) & h(\bm{\nu}_1^T  \bm{\nu}_2 + \beta) & \cdots & h(\bm{\nu}_1^T  \bm{\nu}_K)  & h(\bm{\nu}_1^T  \bm{\nu}_K  + \beta) \cr
                \tau=2;Z=0 & h(\bm{\nu}_2^T  \bm{\nu}_1 + \beta) & h(\bm{\nu}_2^T  \bm{\nu}_1)  & h(\bm{\nu}_2^T \bm{\nu}_2  + \beta) & h(\bm{\nu}_2^T  \bm{\nu}_2) & \cdots & h(\bm{\nu}_2^T  \bm{\nu}_K  + \beta) & h(\bm{\nu}_2^T \bm{\nu}_K) \cr                
                \tau=2;Z=1 & h(\bm{\nu}_2^T  \bm{\nu}_1) & h(\bm{\nu}_2^T \bm{\nu}_1 + \beta) & h(\bm{\nu}_2^T  \bm{\nu}_2) & h(\bm{\nu}_2^T \bm{\nu}_2 + \beta) & \cdots & h(\bm{\nu}_2^T \bm{\nu}_K) & h(\bm{\nu}_2^T \bm{\nu}_K  + \beta )\cr
                \vdots &  \vdots & \vdots & & & \ddots \cr
                \tau=K;Z=0 & h(\bm{\nu}_K^T  \bm{\nu}_1 + \beta) & h(\bm{\nu}_K^T  \bm{\nu}_1)  & h(\bm{\nu}_K^T \bm{\nu}_2  + \beta ) & h(\bm{\nu}_K^T \bm{\nu}_2) & \cdots & h(\bm{\nu}_K^T \bm{\nu}_K  + \beta) & h(\bm{\nu}_K^T  \bm{\nu}_K) \cr                
                \tau=K;Z=1 & h(\bm{\nu}_K^T  \bm{\nu}_1) & h(\bm{\nu}_K^T  \bm{\nu}_1 + \beta) & h(\bm{\nu}_K^T  \bm{\nu}_2) & h(\bm{\nu}_K^T  \bm{\nu}_2 + \beta) & \cdots & h(\bm{\nu}_K^T  \bm{\nu}_K)  & h(\bm{\nu}_K^T   \bm{\nu}_K  + \beta ) \cr 
 } .
\end{equation}
\end{tiny}

So we know that $\beta = h^{-1}(\bm{\theta}_{Z,11}) -  h^{-1}(\bm{\theta}_{Z,12})$, and we can use the estimator
\begin{equation}
    \hat{\beta} = h^{-1}(\widehat{\bm{\theta}}_{Z,11}) -  h^{-1}(\widehat{\bm{\theta}}_{Z,12}).
\end{equation}

By Theorem \ref{lemma:thm2TangEtAl2017} we know that
\begin{eqnarray}
n(\bm{\widehat{\theta}}_{Z,11} -\bm{\theta}_{Z,11} ) & \overset{d}{\rightarrow} & N( \psi_{11},\sigma_{11}^2) \\
n (\bm{\widehat{\theta}}_{Z,12} - \bm{\theta}_{Z,12}) & \overset{d}{\rightarrow} & N( \psi_{12},\sigma_{12}^2)
\end{eqnarray}
Applying a Taylor expansion reveals
\begin{eqnarray}
    h^{-1}(\widehat{\bm{\theta}}_{Z,11}) &=& h^{-1}(\bm{\theta}_{Z,11}) + \left(h^{-1}(\bm{\theta}_{Z,11}) \right)^{\prime} \left(\widehat{\bm{\theta}}_{Z,11}- \bm{\theta}_{Z,11}\right) + \text{smaller order terms} \\
    h^{-1}(\widehat{\bm{\theta}}_{Z,12}) &=& h^{-1}(\bm{\theta}_{Z,12}) + \left(h^{-1}(\bm{\theta}_{Z,12}) \right)^{\prime} \left(\widehat{\bm{\theta}}_{Z,12}- \bm{\theta}_{Z,12}\right) + \text{smaller order terms}
\end{eqnarray}

and this implies
\begin{eqnarray}
n\left(h^{-1}(\bm{\widehat{\theta}}_{Z,11}) -h^{-1}(\bm{\theta}_{Z,11}) \right) & \overset{d}{\rightarrow} & N( \widetilde{\psi}_{11},\widetilde{\sigma}_{11}^2), \\
n \left(h^{-1}(\bm{\widehat{\theta}}_{Z,12}) - h^{-1}(\bm{\theta}_{Z,12}) \right)& \overset{d}{\rightarrow} & N( \widetilde{\psi}_{12},\widetilde{\sigma}_{12}^2),
\end{eqnarray}
where $\widetilde{\psi}_{k\ell} = \psi_{k \ell} \left(h^{-1}(\bm{\theta}_{Z,k\ell}) \right)^{\prime}$ and 
$\widetilde{\sigma}_{k\ell}^2 = \sigma_{k\ell}^2 \left[\left(h^{-1}(\bm{\theta}_{Z,k\ell}) \right)^{\prime}\right]^2$ for $k,\ell=1,2$. 

Finally this implies that our estimator behaves in the manner
\begin{eqnarray}
n(\hat{\beta} - \beta) &\overset{d}{\rightarrow} & N(\widetilde{\psi}_{\beta},\widetilde{\sigma}_{\beta}^2),   
\end{eqnarray}
where the bias term is 
\begin{eqnarray}
\widetilde{\psi}_{\beta}  =\widetilde{\psi}_{11}  -  \widetilde{\psi}_{12},
\end{eqnarray}
where the variance term $\widetilde{\sigma}_{\beta}^2$ satisfies
\begin{eqnarray}
\widetilde{\sigma}_{\beta}^2 = \widetilde{\sigma}_{11}^2+\widetilde{\sigma}_{12}^2 - 2\widetilde{\sigma}_{11,12},
\end{eqnarray}
and where the covariance term is given by
\begin{equation}
   \widetilde{\sigma}_{11,12} =  cov(h^{-1}(\widehat{\bm{\theta}}_{11}), h^{-1}(\widehat{\bm{\theta}}_{12})) = \sigma_{11,12} \left[\left(h^{-1}(\bm{\theta}_{Z,11}) \right)^{\prime}\right]\left[\left(h^{-1}(\bm{\theta}_{Z,12}) \right)^{\prime}\right].
\end{equation}

\end{proof}


\subsection{Proof of THEOREM \ref{thm:clt_general_h_block_estim}}
 Let $K$ be known. 
 When $\bm{\tau}$ is not known, we can estimate it using Adjacency Spectral Embedding (ASE) to get an estimate $\widehat{\bm{Y}}= \widehat{\bm{U}}\vert \widehat{\bm{S}} \vert^{1/2}$; we then cluster the rows of $\widehat{\bm{Y}}$  using a Gaussian Mixture Model (GMM) or $K$-means clustering. This gives estimates $\widehat{\bm{\xi}}$ and therefore $\widehat{\bm{\tau}}$ as we can easily estimate the probability $b_k$ of the Bernoulli covariates from the observed $\bm{Z}_i$'s, within each estimated block.
 
 Because of the estimation the blocks are recovered up to a permutation of the blocks' labels (as it is in any mixture model). Nonetheless, we can still obtain an asymptotic result for $\widehat{\beta}$ using the following Lemma \ref{lemma:cor2TangEtAl2017}, taken from \cite{TangEtAl2017}.

\begin{lemma} (Corollary~2 in \cite{TangEtAl2017}) Let the setting and notation be as in Theorem \ref{lemma:thm2TangEtAl2017}. Let $K$ be known and let $\widehat{\bm{\xi}}: [n]\rightarrow [\tilde{K}]$ be the function that assigns nodes to clusters, estimated using GMM or K-means clustering on the rows of $\widehat{\bm{Y}}= \widehat{\bm{U}}\vert \widehat{\bm{S}} \vert^{1/2}$ (as in the proof of Theorem~\ref{lemma:thm2TangEtAl2017}). Let 
$\widehat{\bm{\theta}}_{Z,k\ell}=\widehat{\bm{\mu}}_{k}^{T} \mathbf{I}_{d_1,d_2} \widehat{\bm{\mu}}_\ell$ and let $\widehat{\bm{\Delta}}=\sum_{k=1}^{K} \widehat{\eta}_k  \widehat{\bm{\mu}}_k \widehat{\bm{\mu}}_{\ell}^{T}$. For $k\in[K]$ and $\ell\in[K]$ define $\widehat{\psi}_{k\ell}$ as

\begin{eqnarray}
\widehat{\psi}_{k\ell} &=& \sum_{r=1}^{K}\widehat{\xi}_r \left(\widehat{\bm{\theta}}_{kr}(1-\widehat{\bm{\theta}}_{kr}) + \widehat{\bm{\theta}}_{\ell r}(1-\widehat{\bm{\theta}}_{\ell r})\right) \widehat{\bm{\mu}}_k^T \widehat{\Delta}^{-1} \bm{I}_{d_1,d_2} \widehat{\Delta}^{-1} \widehat{\bm{\mu}}_\ell \\
  &-& \sum_{r=1}^{K}\sum_{s=1}^{K} \widehat{\eta}_r \widehat{\eta}_s
\widehat{\bm{\theta}}_{sr} (1- \widehat{\bm{\theta}}_{sr})  \widehat{\bm{\mu}}_s^T \widehat{\Delta}^{-1} \bm{I}_{d_1,d_2} \widehat{\Delta}^{-1}(\widehat{\bm{\mu}}_\ell \widehat{\bm{\mu}}_k^T + \widehat{\bm{\mu}}_k \widehat{\bm{\mu}}_{\ell}^{T} )\widehat{\Delta}^{-1}\widehat{\bm{\mu}}_s
\end{eqnarray}
Then there exists a sequence of permutations $\phi\equiv \phi_n$ on $[K]$ such that for any $k\in[K]$ and $\ell\in[K]$, 
\begin{equation}
    n\left(\widehat{\bm{\theta}}_{\phi(k),\phi(\ell)}-\bm{\theta}_{k\ell}-\frac{\widehat{\psi}_{k\ell}}{n}\right)\overset{d}{\rightarrow} N(0, \sigma_{k\ell}^2) 
\end{equation}
as $n\rightarrow\infty$.
\label{lemma:cor2TangEtAl2017}
\end{lemma}
\begin{proof}
See \cite{TangEtAl2017} for the detailed proof.
\end{proof}

Let us first focus on the \textbf{linear} case, in which $h$ is the identity function. If $h(u)=u$, then we can estimate $\beta$ as

\begin{eqnarray}
\widehat{\beta} &=& \widehat{\bm{\theta}}_{\phi(1),\phi(1)} -\widehat{\bm{\theta}}_{\phi(1),\phi(2)} \\
&=& \widehat{\bm{\theta}}_{\phi(1),\phi(1)} - \bm{\theta}_{11} + \bm{\theta}_{11} - \bm{\theta}_{12} + \bm{\theta}_{12} -\widehat{\bm{\theta}}_{\phi(1),\phi(2)} \\
&=& \left(\widehat{\bm{\theta}}_{\phi(1),\phi(1)} - \bm{\theta}_{11}\right) + \beta   -\left(\widehat{\bm{\theta}}_{\phi(1),\phi(2)} -\bm{\theta}_{12}\right)
\end{eqnarray}
and rearranging we obtain
\begin{equation}
    n\left(\widehat{\beta} - \beta\right) = n\left(\widehat{\bm{\theta}}_{\phi(1),\phi(1)} - \bm{\theta}_{11}\right)    -n\left(\widehat{\bm{\theta}}_{\phi(1),\phi(2)} -\bm{\theta}_{12}\right)
\end{equation}
which by Lemma \ref{lemma:cor2TangEtAl2017} implies that there exists a (sequence of) permutation(s) $\phi$ such that 

\begin{equation}
    n\left(\widehat{\beta} - \beta  -  \frac{\widehat{\psi}_{\beta}}{n}\right) \overset{d}{\rightarrow} N(0, \sigma_\beta^2 )
\end{equation}
where $\widehat{\psi}_\beta =(\widehat{\psi}_{11}-\widehat{\psi}_{12})$ and  $\sigma_\beta^2= \sigma_{11}^2 + \sigma_{12}^2 - 2cov\left(\widehat{\bm{\theta}}_{\phi(1),\phi(1)},\widehat{\bm{\theta}}_{\phi(1),\phi(2)}\right).$

For the \textbf{nonlinear} link function, we use a Taylor expansion
\begin{eqnarray}
    h^{-1}(\widehat{\bm{\theta}}_{\phi(k),\phi(\ell)}) &=& h^{-1}(\bm{\theta}_{k\ell}) + [h^{-1}(\bm{\theta}_{k\ell})]^{\prime} \left( \widehat{\bm{\theta}}_{\phi(k),\phi(\ell)}-\bm{\theta}_{k\ell} \right) + \text{smaller order terms}
\end{eqnarray}
and thus we have
\begin{eqnarray}
    h^{-1}(\widehat{\bm{\theta}}_{\phi(k),\phi(\ell)}) - h^{-1}(\bm{\theta}_{k\ell}) &=& [h^{-1}(\bm{\theta}_{k\ell})]^{\prime} \left( \widehat{\bm{\theta}}_{\phi(k),\phi(\ell)}-\bm{\theta}_{k\ell} \right) + \text{smaller order terms}
\end{eqnarray}
which implies that 
\begin{eqnarray}
n\left(h^{-1}(\widehat{\bm{\theta}}_{\phi(k),\phi(\ell)}) - h^{-1}(\bm{\theta}_{k\ell}) \right) \overset{d}{\rightarrow} N\left(\tilde{\tilde{\psi}}_{k\ell}, \tilde{\tilde{\sigma}}_{k\ell}^2 \right),
\end{eqnarray}
where $\tilde{\tilde{\psi}}_{k\ell}=[h^{-1}(\bm{\theta}_{k\ell})]^{\prime} \widehat{\psi}_{k\ell}$ and $\tilde{\tilde{\sigma}}_{k\ell}^2 =([h^{-1}(\bm{\theta}_{k\ell})]^{\prime})^2 \sigma_{k\ell}^2 $.

We therefore obtain the result 
\begin{equation}
    n(\widehat{\beta} -\beta)  \overset{d}{\rightarrow} N(\widehat{\psi}_\beta, \widehat{\sigma}_\beta^2)
\end{equation}
where $\widehat{\psi}_\beta = \tilde{\tilde{\psi}}_{11} -  \tilde{\tilde{\psi}}_{12}$ and 
$\widehat{\sigma}_\beta^2 = \tilde{\tilde{\sigma}}_{11}^2 + \tilde{\tilde{\sigma}}_{12}^2 -2\tilde{\tilde{\sigma}}_{11,12}$
and 
\begin{eqnarray}
\tilde{\tilde{\sigma}}_{11,12} : = cov\left(h^{-1}(\widehat{\bm{\theta}}_{\phi(1),\phi(1)}),h^{-1}(\widehat{\bm{\theta}}_{\phi(1),\phi(2)})\right) =  \sigma_{11,12}\left[\left(h^{-1}(\widehat{\bm{\theta}}_{\phi(1)\phi(1)})\right)' \right] \left[\left(h^{-1}(\widehat{\bm{\theta}}_{\phi(1),\phi(2)})\right)' \right].
\end{eqnarray}

\subsection{Proof of THEOREM \ref{thm:clt_general_h_semisparse}}
In the semi-sparse regime the proof follows the same steps as the proof of Theorem~\ref{thm:clt_general_h_block_estim}. The only difference is in the bias terms, variances, and covariance terms, that are computed according to the following lemma.

\begin{lemma} (Corollary 2 extension for semi-sparse case in \cite{TangEtAl2017}).\\
Let $\bm{A}\sim SBM(\bm{\xi},\bm{\theta},\rho_n)$ be a $\widetilde{K}$-block stochastic blockmodel adjacency matrix on $n$ vertices with sparsity factor $\rho_n$. Let $\mu_1, \dots ,\mu_{\widetilde{K}}$ be the center of the blocks in $\mathbb{R}^d$ and let $\bm{\theta}_{k\ell}=\bm{\mu}_k^T \bm{I}_{d_1,d_2} \bm{\mu}_{\ell}$ be the probability of a link between nodes in blocks $k$ and $\ell$. 
Let $K$ be known and let $\widehat{\bm{\xi}}: [n]\rightarrow [\tilde{K}]$ be the function that assigns nodes to clusters, estimated using GMM or K-means clustering on the rows of $\widehat{\bm{Y}}= \widehat{\bm{U}}\vert \widehat{\bm{S}} \vert^{1/2}$ (as in the proof of Theorem~\ref{lemma:thm2TangEtAl2017}). Let 
$\widehat{\bm{\theta}}_{k\ell}=\widehat{\bm{\mu}}_{k}^{T} \mathbf{I}_{d_1,d_2} \widehat{\bm{\mu}}_\ell$ and let $\widehat{\bm{\Delta}}=\sum_{k=1}^{K} \widehat{\eta}_k  \widehat{\bm{\mu}}_k \widehat{\bm{\mu}}_{\ell}^{T}$.
Define $\Delta = \sum_{k=1}^{K} \eta_k \bm{\mu}_k \bm{\mu}_k^T$ and
let $\zeta_{k\ell}=\bm{\mu}_k^T \Delta^{-1}\bm{\mu}_\ell $. Define $\widetilde{\sigma}_{kk}^2$ for $k\in[K]$ to be
\begin{eqnarray}
\widetilde{\sigma}_{kk}^2  &=& 4 \bm{\theta}_{kk} \zeta_{kk}^2 +
4\sum_{r=1}^K \eta_r\bm{\theta}_{kr} \zeta_{kr}^2 \left( \frac{1}{\eta_k} -2 \zeta_{kk}  \right)^2 
+ 2 \sum_{r=1}^K \sum_{s=1}^K  \eta_r \eta_s \bm{\theta}_{rs}\zeta_{kr}^2 \zeta_{ks}^2
\end{eqnarray}
and define $\widetilde{\sigma}_{k\ell}^2$ for $k\in[\tilde{K}]$ and $\ell\in [\tilde{K}]$, $ k \neq \ell$ to be
\begin{eqnarray}
\widetilde{\sigma}_{k\ell}^2 &=& 
\left( \bm{\theta}_{kk} + \bm{\theta}_{\ell\ell} \right) \zeta_{kk}^2 + 2 \bm{\theta}_{k\ell} \zeta_{kk}\zeta_{\ell\ell} 
+ \sum_{r=1}^{K} \eta_r \bm{\theta}_{kr}\zeta_{\ell r}^2 \left(\frac{1}{\eta_k} - 2\zeta_{kk} \right) \\
&+& \sum_{r=1}^{K} \eta_r \bm{\theta}_{\ell r}\zeta_{k r}^2 \left(\frac{1}{\eta_\ell} - 2\zeta_{\ell} \right) 
- 2 \sum_{r=1}^K \eta_r \left( \bm{\theta}_{kr} + \bm{\theta}_{\ell r} \right) \zeta_{kr}\zeta_{r \ell}\zeta_{k\ell} \\
&+& \frac{1}{2}\sum_{r=1}^K \sum_{s=1}^K \eta_r \eta_s \bm{\theta}_{rs}\left(\zeta_{kr}\zeta_{\ell s} + \zeta_{\ell r}\zeta_{ks}\right)^2 .
\end{eqnarray}

Let $\ddot{\psi}_{k\ell}$ be defined as
\begin{eqnarray}
\ddot{\psi}_{k\ell} &=& \sum_{r=1}^{K} \widehat{\eta}_r
\left(\widehat{\bm{\theta}}_{kr}+\widehat{\bm{\theta}}_{\ell r}\right)\widehat{\bm{\mu}}_k^T \widehat{\Delta}^{-1}\widehat{\bm{\mu}}_\ell \\
&-& \sum_{r=1}^{K}\sum_{s=1}^{K} \widehat{\bm{\mu}}_k \widehat{\bm{\mu}}_s \widehat{\bm{\theta}}_{sr}\widehat{\bm{\mu}}_s^T \widehat{\Delta}^{-1}\mathbf{I}_{d_1,d_2}\widehat{\Delta}^{-1}\left(\widehat{\bm{\mu}}_\ell \widehat{\bm{\mu}}_k^T + \widehat{\bm{\mu}}_l \widehat{\bm{\mu}}_\ell^T\right)\widehat{\Delta}^{-1}\widehat{\bm{\mu}}_s .
\end{eqnarray}
Then there exists a (sequence of) permutation(s) $\phi\equiv \phi_n$ on $[K]$ such that for any $k\in[K]$ and $\ell\in[K]$, 
\begin{equation}
    n\rho_n^{1/2}\left( \widehat{\bm{\theta}}_{\phi(k),\phi(\ell)} - \bm{\theta}_{k\ell} - \frac{\ddot{\psi}_{k\ell}}{n\rho_n}  \right) \overset{d}{\rightarrow} N(0,\widetilde{\sigma}_{k\ell})
\end{equation}
as $n\rightarrow \infty$, $\rho_n\rightarrow 0$, and $n\rho_n = \omega(\sqrt{n})$.
\end{lemma}

\end{document}